\documentclass{llncs}

\input xy

\xyoption{all}

\usepackage{hyperref}
\usepackage{amsmath}
\usepackage{amssymb}
\usepackage{stmaryrd}
\usepackage{todonotes}

\newcommand{\ms}[1]{} %%{\todo[inline]{MS: #1}} 
\newcommand{\pt}[1]{} %%{\todo[inline]{PT: #1}} 

\newcommand{\ba}{\begin{array}}
\newcommand{\ea}{\end{array}}
\newcommand{\bda}{\[\ba}
\newcommand{\eda}{\ea\]}

\newcommand{\nextL}[1]{\bigcirc \, #1}
\newcommand{\alwaysL}[1]{\Box\, #1}
\newcommand{\eventuallyL}[1]{\Diamond\, #1}
\newcommand{\untilL}[2]{#1\, \mathbf{U} \, #2}
\newcommand{\releaseL}[2]{#1\, \mathbf{R} \, #2}
\newcommand{\true}{\ensuremath{\mathbf{tt}}}
\newcommand{\false}{\ensuremath{\mathbf{ff}}}

\newcommand{\markL}[1]{#1*}

\newcommand\A{\mathcal{A}}

\renewcommand\L{\mathcal{L}}

\newcommand\Power{\mathbb{P}}
\newcommand\nat{\omega}%{\mathbb{N}}

\newcommand\pderiv[2]{\partial_{#2}(#1)}
\newcommand\partDeriv[2]{{\mathit pd}_{#2}(#1)}

\newcommand\Angle[1]{\langle#1\rangle}
\newcommand\LF{\textsc{lf}}

\newcommand\AP{\mathbf{AP}}
\newcommand\INT{\mathcal{I}}

\newcommand\otop{{\textcircled{\raisebox{-0.5ex}{$\top$}}}}

\newcommand\pbfa\alpha
\newcommand\pbfb\beta

\newcommand\forma\varphi
\newcommand\formb\psi
\newcommand\formc\xi

\newcommand\symx{x}
\newcommand\symy{y}

\newcommand\scm\odot % smart conjunction of monomials

\newcommand\SET{\ensuremath{\mathcal{S}}}
\newcommand\SIMP{\ensuremath{\mathcal{T}}}

\newcommand\FORMA\Phi
\newcommand\FORMB\Psi

\newcommand\monoa\mu
\newcommand\monob\nu

\newcommand\Lit\ell

\newcommand{\sem}[1]{[\![#1]\!]}

\def\hyph{-\penalty0\hskip0pt\relax}

\newcommand{\myirule}[2]{{\renewcommand{\arraystretch}{1.2}\ba{c} #1
                      \\ \hline #2 \ea}}

\newcommand{\rlabel}[1]{\mbox{(#1)}}

\renewcommand{\owedge}{\wedge}

\renewcommand{\otop}{\true}

\begin{document}
\title{LTL Semantic Tableaux and Alternating $\omega$-automata via Linear Factors}

\author{Martin Sulzmann\inst1  \and Peter Thiemann\inst2}
 \institute{
   % Faculty of Computer Science and Business Information Systems \\
   Karlsruhe University of Applied Sciences \\
   % Moltkestrasse 30, 76133 Karlsruhe, Germany\\
   \email{martin.sulzmann@hs-karlsruhe.de}
   \and 
   Faculty of Engineering, University of Freiburg\\
   % Georges-K{\"o}hler-Allee  079, 79110 Freiburg, Germany \\
   \email{thiemann@acm.org}
 }

\maketitle
\begin{abstract}
    Linear Temporal Logic (LTL) is a widely used specification framework for linear time properties of
  systems. The standard approach for verifying such properties is by transforming LTL formulae to suitable
  $\omega$-automata and then applying model checking. 
  We revisit Vardi's transformation of an LTL formula to an alternating $\omega$-automaton
  and Wolper's LTL tableau method for satisfiability checking.
  We observe that both constructions effectively rely on a decomposition of formulae
  into \emph{linear factors}.
  Linear factors have been introduced previously by Antimirov
  in the context of regular expressions.
  We establish the notion of linear factors for LTL and verify essential properties
  such as expansion and finiteness.
  Our results shed new insights on the connection
  between the construction of alternating $\omega$-automata
  and semantic tableaux.
\end{abstract}

%-------------------------------------------------------------------------
%-------------------------------------------------------------------------
\section{Introduction}
\label{sec:introduction}

Linear Temporal Logic (LTL) is a widely used specification framework
for linear time properties of systems. An LTL formula describes a
property of an infinite trace of a system. Besides the usual logical
operators, LTL supports the temporal operators $\nextL{\forma}$
($\forma$ holds in the next step of the trace) and
$\untilL\forma\formb$ ($\forma$ holds for all steps in the trace until
$\formb$ becomes true). LTL can describe many relevant safety and liveness properties.

The standard approach to verify a system against an LTL formula is
model checking. To this end, the verifier translates a formula into a
suitable $\omega$\hyph automaton, for example, a B\"uchi automaton or an
alternating automaton, and applies the model checking algorithm to the
system and the automaton. This kind of translation is widely studied
because it is the enabling technology for model checking
\cite{DBLP:conf/focs/WolperVS83,DBLP:journals/iandc/VardiW94,DBLP:conf/lics/VardiW86}.
Significant effort is spent on developing 
translations that generate (mostly) deterministic automata or that minimize the
number of states in the generated automata \cite{DBLP:conf/cav/GastinO01,DBLP:conf/tacas/BabiakKRS12}. Any
improvement in these dimensions is valuable as it speeds up the model
checking algorithm.

Our paper presents a new approach to understanding and proving the
correctness of Vardi's construction of alternating automata (AA)  from
LTL formulae~\cite{Vardi:1997:AAU:648233.753439}. Our approach is
based on a novel adaptation to LTL of 
linear factors, a concept from
Antimirov's construction of partial derivatives of regular
expressions~\cite{Antimirov96Partial}.  
Interestingly, a similar construction yields a new explanation
for Wolper's construction of semantic
tableaux~\cite{wolper85-ltl-tableau} for checking satisfiability of LTL formulae, thus
hinting at a deep connection between the two constructions.

The paper contains the following contributions.
\begin{itemize}
\item Definition of linear factors and partial derivatives for LTL
  formulae (Section~\ref{sec:linear-factors}).
   We establish their properties and prove correctness.
\item Transformation from LTL to AA based on linear factors. The
  resulting transformation is essentially the standard LTL to AA
  transformation \cite{Vardi:1997:AAU:648233.753439}; it is
  correct by construction of the linear factors (Section~\ref{sec:construction}).
\item Construction of semantic tableaux to determine satisfiability of
  LTL formulae using linear factors (Section~\ref{sec:semantic-tableau}). Our
  method corresponds closely to Wolper's construction and comes 
  with a free correctness proof.
\end{itemize}  

Proofs for results stated can be found in the appendix.

\subsection{Preliminaries}
\label{sec:preliminaries}

We write $\nat = \{0, 1, 2, \dots\}$ for the set of natural numbers
and $\Sigma^\omega$ for the set of infinite words over alphabet 
$\Sigma$ with symbols ranged over by $\symx,\symy \in \Sigma$. We regard a word
$\sigma\in\Sigma^\omega$ as a 
map and write $\sigma_i$ ($i\in\nat$) for the $i$-th symbol. For
$n\in\nat$, we write $\sigma[n\dots]$ for the 
suffix of $\sigma$ starting at $n$, that is, the function $i \mapsto
\sigma_{n+i}$, for $i\in\nat$.  We write $\symx\sigma$ for prepending
symbol $\symx$ to $\sigma$, that is, $(\symx\sigma)_0 = \symx$ and
$(\symx\sigma)_{i+1} = \sigma_i$, for all $i\in\nat$.
The notation $\Power (X)$ denotes the power set of $X$.

%-------------------------------------------------------------------------
%-------------------------------------------------------------------------
\section{Linear Temporal Logic}
\label{sec:ltl}

Linear temporal logic (LTL)~\cite{Pnueli:1977:TLP:1382431.1382534}
enhances propositional logic with the temporal operators $\nextL{\forma}$ ($\forma$ will be true in
the next step) and $\untilL\forma\formb$ ($\forma$ holds until $\formb$ becomes true). LTL formulae $\forma,\formb$ are
defined accordingly where we draw atomic propositions $p, q$ from a finite set $\AP$.
\begin{definition}[Syntax of LTL]
\begin{align*}
  \forma, \formb &::= p \mid \true \mid \neg\forma \mid \forma\wedge\formb \mid \nextL\forma \mid \untilL\forma\formb
\end{align*}
\end{definition}
We apply standard precedence rules to parse a formula ($\neg$,
$\nextL\forma$, and other prefix operators bind strongest; then
$\untilL\forma\formb$ and the upcoming $\releaseL\forma\formb$
operator; then 
conjunction and finally disjunction with the weakest binding
strength; as the latter are associative, we do not group their
operands explicitly). We use parentheses to deviate from precedence or to
emphasize grouping of subformulae.

A model of an LTL formula is an infinite word $\sigma \in \Sigma^\omega$ where
$\Sigma$ is a finite set of symbols and there is an interpretation $ \INT :
\Sigma \to \Power (\AP)$ that maps a symbol $x\in\Sigma$ to the finite set of atomic predicates that are true for $x$.

\begin{definition}[Semantics of LTL]
  The formula $\forma$ holds on word $\sigma \in \Sigma^\omega$ as
  defined by the judgment $\sigma  \models \forma$. 
\begin{align*}
 \sigma \models p
  &\Leftrightarrow p \in \INT (\sigma_0)
  \\
  \sigma \models \true
  &
  \\
  \sigma \models \neg \forma
  &\Leftrightarrow \sigma \not\models \forma
  \\
  \sigma \models \forma \wedge \formb
  & \Leftrightarrow \sigma \models \forma \text{ and } \sigma \models \formb
  \\
  \sigma \models \nextL{\forma}
  & \Leftrightarrow \sigma[1\dots] \models \forma
  \\
  \sigma \models \untilL{\forma}{\formb}
  & \Leftrightarrow \exists n \in \nat,
    (\forall j \in \nat, j<n \Rightarrow \sigma[j\dots] \models \forma)
    \text{ and } \sigma[n\dots] \models \formb
\end{align*}
We say $\forma$ is \emph{satisfiable} if there exists $\sigma \in \Sigma^\omega$ such that $\sigma \models \forma$.
\end{definition}

\begin{definition}[StandardDerived LTL Operators]
%There are several standard derived operators:
\begin{align*}
  \false &= \neg\true \\
  \forma\vee\formb& = \neg (\neg\forma\wedge \neg\formb) & \text{disjunction} \\
  \releaseL\forma\formb & = \neg( \untilL{\neg\forma}{\neg\formb}) & \text{release} \\
  \eventuallyL\formb & = \untilL\true\formb & \text{eventually/finally} \\
  \alwaysL\formb & = \releaseL\false\formb & \text{always/globally}
\end{align*}
\end{definition}

For many purposes, it is advantageous to restrict LTL formulae to positive normal form (PNF). 
In PNF negation only occurs adjacent to atomic propositions. Using the
derived operators, all negations can be pushed inside by using the de Morgan laws. Thanks to the release operator, this transformation runs in linear time and space.
The resulting grammar of formulae in PNF is as follows.

\begin{definition}[Positive Normal Form]
\bda{rcll}
\forma, \formb  & ::= & p \mid \neg p \mid \true \mid \false
\mid  \forma \wedge \formb \mid \forma \vee \formb
\mid \nextL{\forma} \mid \untilL{\forma}{\formb} \mid \releaseL \forma \formb
\eda
\end{definition}

From now on, we assume that all LTL formulae are in PNF.

% \textbf{PT:} Meanwhile I think one could alternatively proceed along the lines of your original
% suggestion where the ``steps'' in the linear factors are finite conjunctions of atomic
% (negated) predicates. As you remarked, this choice enables reasoning with infinite alphabets and (I think) it
% simplifies some proof steps. --- I was really interested to get something analogous to the textbook proofs.

We make use of several standard equivalences in LTL.
\begin{theorem}[Standard results about LTL]~\\[-\baselineskip]
\label{th:standard-ltl}  
  \begin{enumerate}
  % \item $\nextL{\eventuallyL{\forma}} \Leftrightarrow \eventuallyL{\nextL{\forma}}$
  % \item $\nextL{\alwaysL{\forma}} \Leftrightarrow \alwaysL{\nextL{\forma}}$
  \item $\nextL{(\forma \wedge \formb)} \Leftrightarrow (\nextL{\forma}) \wedge (\nextL{\formb})$
  \item $\nextL{(\forma \vee \formb)} \Leftrightarrow (\nextL{\forma}) \vee (\nextL{\formb})$
  % \item $(\alwaysL\forma \wedge \alwaysL\formb) \Leftrightarrow \alwaysL{(\forma\wedge\formb)}$
  % \item $(\eventuallyL\forma \vee \eventuallyL\formb) \Leftrightarrow \eventuallyL{(\forma\vee\formb)}$
  % \item $\eventuallyL\forma \Leftrightarrow \forma \vee \nextL{(\eventuallyL\forma)}$
  % \item $\alwaysL\forma \Leftrightarrow \forma \wedge \nextL{(\alwaysL\forma)}$
  \item $\untilL{\forma}{\formb} \Leftrightarrow
    \formb \vee (\forma \wedge \nextL{(\untilL{\forma}{\formb})})$
  \item $\releaseL{\forma}{\formb} \Leftrightarrow
    \formb \wedge( \forma
    \vee \nextL{(\releaseL{\forma}{\formb})})$
  % \item $\eventuallyL\forma \Leftrightarrow \untilL\true\forma$
  % \item $\alwaysL\forma \Leftrightarrow \releaseL\false\forma$
  % \item $\neg(\untilL\forma\formb) \Leftrightarrow \releaseL{\neg\forma}{\neg\formb}$
  % \item $\neg(\releaseL\forma\formb) \Leftrightarrow \untilL{\neg\forma}{\neg\formb}$
  \end{enumerate}
\end{theorem}
We also make use of the direct definition of a model for the release operation.
\begin{lemma}
  $\sigma \models \releaseL\forma\formb$ is equivalent to one of the following:
  \begin{trivlist}
  \item
    $ \forall n \in \nat, ( \sigma[n\dots] \models \formb \text{ or } \exists j \in \nat, ((j<n) \wedge \sigma[j\dots]
    \models \forma)) $
  \item $\forall n \in \nat, \sigma[n\dots] \models \formb \text{ or }
    \exists j \in \nat, \sigma[j\dots] \models
    \forma \text{ and } \forall i \in \nat, i\le j \Rightarrow \sigma[i\dots] \models \formb$.
  \end{trivlist}
\end{lemma}
% \begin{proof}
%   Standard.
% \begin{align*}
%   \sigma \models \releaseL\forma\formb 
%   &\Leftrightarrow
%   \sigma \models \neg(\untilL{\neg\forma}{\neg\formb})
%   \\
%   &\Leftrightarrow
%     \neg
%     \exists n \in \nat, (\sigma[n\dots] \models \neg\formb \text{ and
%     }  \forall j \in \nat, j<n \Rightarrow \sigma[j\dots] \models \neg\forma 
%     \\
%   &\Leftrightarrow
%     \forall n \in \nat,
%     \neg 
%     ( \sigma[n\dots] \models \neg\formb
%      \text{ and }
%     \forall j \in \nat, j<n \Rightarrow \sigma[j\dots] \models
%     \neg\forma)
%     \\
%   &\Leftrightarrow
%     \forall n \in \nat,
%     ( \sigma[n\dots] \models \formb
%     \text{ or }
%     \neg
%     \forall j \in \nat, j<n \Rightarrow \sigma[j\dots] \models
%     \neg\forma)
%     \\
%   &\Leftrightarrow
%     \forall n \in \nat,
%     ( \sigma[n\dots] \models \formb
%     \text{ or }
%     \exists j \in \nat,
%     \neg 
%     (j<n \Rightarrow \sigma[j\dots] \models \neg\forma))
%     \\
%   &\Leftrightarrow
%     \forall n \in \nat,
%     ( \sigma[n\dots] \models \formb
%     \text{ or }
%     \exists j \in \nat,
%     \neg 
%     (\neg (j<n) \vee \sigma[j\dots] \models \neg\forma))
%     \\
%   &\Leftrightarrow
%     \forall n \in \nat,
%     ( \sigma[n\dots] \models \formb
%     \text{ or }
%     \exists j \in \nat,
%     ((j<n) \wedge  \sigma[j\dots] \models \forma))
% \end{align*}
% \end{proof}

%-------------------------------------------------------------------------
%-------------------------------------------------------------------------
\section{Linear Factors and Partial Derivatives}
\label{sec:linear-factors}

Antimirov~\cite{Antimirov96Partial} defines a linear factor of a
regular expression as a pair of an input symbol and a next formula (regular expression).
The analogue for LTL is a pair $\Angle{\monoa,\forma}$ where $\monoa$ is a propositional formula in monomial form (no
modalities, see Definition~\ref{def:monomial}) that models the set of
first symbols whereas $\forma$ is a formal conjunction of temporal LTL 
formulae for the rest of the input. Informally, 
$\Angle{\monoa, \forma}$ corresponds to $\monoa \wedge \nextL{\forma}$. A formula always gives rise to a set of linear
factors, which is interpreted as their disjunction.

% \begin{definition}
%   Let $<$ be an arbitrary, fixed, total ordering on $ \AP$.
%   A \emph{monomial formula} has the form $p_1^{i_1} \wedge \dots
%   \wedge p_n^{i_n}$ where $ \{p_1, \dots, p_n\} \subseteq \AP$, for all $1\le i<j\le n$ $p_i < p_j$, and $i_j \in \{+1,
%   -1\} $ with $p^{+1} = p$ and $p^{-1} = \neg p$. We write $\true$ for the monomial with $n=0$
%   atoms. 
%   We let $\monoa$ and $\monob$ range over monomials.
% \end{definition}

\begin{definition}[Temporal Formulae, Literals and Monomials]\label{def:monomial}
A \emph{temporal formula} does not start with a conjunction
or a disjunction.
  
  A \emph{literal} $\Lit$ of $\AP$ is an element of $\AP \cup \neg\AP$. Negation of negative literals is
  defined by $\neg (\neg p) = p$.
  
  A \emph{monomial} $\monoa,\monob$ is either $\false$ or a set of literals of $\AP$ such that
  $\Lit\in\monoa$ implies $\neg \Lit \notin\monoa$.
  The formula associated with a monomial $\monoa$ is given by
  \begin{align*}
    \Theta (\monoa) &=
                      \begin{cases}
                        \false & \monoa = \false \\
                        \bigwedge \monoa & \monoa \text{ is a set of literals.}
                      \end{cases}
  \end{align*}
\end{definition}
In particular, if $\mu = \emptyset$, then $\Theta (\mu) = \true$. Hence, we may write $\true$
for the empty-set monomial.
As a monomial is represented either by $\false$ or by a set of
non-contradictory  literals, its representation is unique.

We define a smart conjunction operator on monomials that retains the monomial structure.

\begin{definition}
  Smart conjunction on monomials is defined as their union unless their conjunction $\Theta (\monoa)
  \wedge \Theta (\monob)$ is equivalent to $\false$.
  $$\monoa \scm \monob =
  \begin{cases}
    \false & \monoa = \false \vee \monob = \false \\
    \false & \exists \Lit \in \monoa\cup\monob.\ \neg \Lit \in \monoa\cup\monob \\
    \monoa \cup\monob & \text{otherwise.}
  \end{cases}
  $$
\end{definition}

Smart conjunction of monomials is correct in the sense that it produces results equivalent to the conjunction
of the associated formulae.
\begin{lemma}\label{lemma:conjunction-monomials}
  $\Theta(\monoa) \wedge \Theta(\monob) \Leftrightarrow \Theta(\monoa \scm \monob)$.
\end{lemma}

We define an operator $\SIMP$ that transforms propositional formulae
consisting of literals and temporal subformulae into sets of conjunctions.
We assume that  conjunction $\owedge$ simplifies formulae to normal
form using associativity, commutativity, and idempotence. The normal
form relies on a total ordering of formulae derived from an
(arbitrary, fixed) total ordering on atomic propositions. 

\begin{definition}[Set-Based Conjunctive Normal Form]
\begin{align*}
  \SIMP (\forma \wedge \formb) &= \{ \forma' \owedge \formb' \mid
                                 \forma' \in \SIMP(\forma), \formb'
                                 \in \SIMP (\formb) \} \\
  \SIMP (\forma \vee \formb) &= \SIMP(\forma) \cup \SIMP (\formb) \\
  \SIMP (\forma) &= \{\forma\} & \text{if $\forma$ is a temporal formula}
\end{align*}
\end{definition}

\begin{lemma}\label{lemma:simp-correct}
  $\bigvee (\SIMP\forma) \Leftrightarrow \forma$. 
\end{lemma}

\begin{definition}[Linear Factors]
  The set of  \emph{linear factors $\LF({\forma})$ of an LTL formula} in PNF is
  defined as a set of pairs of a monomial and a PNF formula in conjunctive
  normal form.
  \bda{lcl}
  \LF (\Lit) & = & \{ \Angle{\{ \Lit \},\otop}\}
  \\
  \LF (\true) &= & \{ \Angle{\true, \otop} \}
  \\
  \LF (\false) &= & \{ \}
  \\
  \LF (\forma \vee \formb) & = & \LF (\forma) \cup \LF (\formb)
  \\
  % \LF (\forma \wedge \formb) & = & \{ \Angle{\monoa \scm
  %   \monob,\forma' \owedge \formb'} \mid \Angle{\monoa,\forma'} \in
  % \LF (\forma), \Angle{\monob,\formb'} \in \LF (\formb) \}
  % \\
  \LF (\forma \wedge \formb) & = & \{ \Angle{\monoa',\forma' \owedge \formb'} \mid \Angle{\monoa,\forma'} \in
  \LF (\forma), \Angle{\monob,\formb'} \in \LF (\formb),
  \monoa' = \monoa \scm \monob \ne \false\}
  \\
  \LF (\nextL{\forma}) & = & \{ \Angle{\true,\forma'} \mid \forma' \in
  \SIMP(\forma)\}
  \\
  \LF (\untilL{\forma}{\formb}) & = & \LF (\formb) \cup \{
  \Angle{\monoa,\forma' \owedge \untilL{\forma}{\formb}} \mid
  \Angle{\monoa,\forma'} \in \LF (\forma) \}
  \\
  % \LF (\releaseL{\forma}{\formb}) & = &
  % \{  \Angle{\monoa \scm \monob, \forma' \owedge \formb'} \mid
  % \Angle{\monoa, \forma'} \in \LF (\forma),
  % \Angle{\monob, \formb'} \in \LF (\formb)
  % \} \\ &\cup& \{
  % \Angle{\monob, \formb' \owedge \releaseL{\forma}{\formb}} \mid
  % \Angle{\monob,\formb'} \in \LF (\formb)\}
  % \\
  \LF (\releaseL{\forma}{\formb}) & = &
  \{  \Angle{\monoa', \forma' \owedge \formb'} \mid
  \Angle{\monoa, \forma'} \in \LF (\forma),
  \Angle{\monob, \formb'} \in \LF (\formb),
  \monoa' = \monoa \scm \monob \ne \false
  \} \\ &\cup& \{
  \Angle{\monob, \formb' \owedge \releaseL{\forma}{\formb}} \mid
  \Angle{\monob,\formb'} \in \LF (\formb)\}
  \eda
\end{definition}
By construction, the first component of a linear factor is never
$\false$. Such pairs are eliminated from the beginning by the tests
for $\monoa \scm \monob \ne \false$.

We can obtain shortcuts for the derived operators ``always'' and ``eventually''.
\begin{lemma}~\\[-2.7\baselineskip]
  \begin{align*}
    \LF (\eventuallyL{\formb}) & =  \LF (\formb) \cup \{
                                     \Angle{\true,\eventuallyL{\formb}} \}
    \\
    \LF (\alwaysL{\formb})
                               & =  \{ \Angle{\monob,\formb' \owedge
                                 \alwaysL{\formb}} \mid \Angle{\monob,\formb'} \in \LF (\formb) \}
  \end{align*}
\end{lemma}
\begin{example}
  Consider the formula $\alwaysL{\eventuallyL{p}}$.
  \begin{align*}
    \LF (\eventuallyL{p})
    &=\LF (p) \cup \{\Angle{\true,\eventuallyL{p}} \}\\
    &= \{ \Angle{p, \otop}, \Angle{\true, \eventuallyL{p}} \}
    \\
    \LF (\alwaysL{\eventuallyL{p}})
    &= \{ \Angle{\monoa,\forma' \owedge
      \alwaysL{\eventuallyL{p}}} \mid \Angle{\monoa,\forma'} \in \LF (\eventuallyL{p}) \} \\
    &= \{ \Angle{\monoa,\forma' \owedge
      \alwaysL{\eventuallyL{p}}} \mid \Angle{\monoa,\forma'} \in \{ \Angle{p, \otop}, \Angle{\true, \eventuallyL{p}} \} \} \\
    &= \{ \Angle{p, \alwaysL{\eventuallyL{p}}}, \Angle{\true, \eventuallyL{p} \owedge \alwaysL{\eventuallyL{p}}} \}
  \end{align*}
\end{example}

\begin{definition}[Linear Forms]
  A formula $\forma = \bigvee_{i\in I} b_i \wedge \nextL{\forma_i}$ is
  in \emph{linear form} if each $b_i$ is a conjunction of literals and
  each $\forma_i$ is a temporal formula. 
  
  The formula associated to a set of linear factors is
  in linear form as given by the following mapping.
  \begin{align*}
    \Theta  ( \{ \Angle{\monoa_i, \forma_i} \mid i \in I \})
    & = \bigvee_{i \in I} (\Theta(\monoa_i) \wedge \nextL{\forma_i})
  \end{align*}
\end{definition}

Each formula can be represented in linear form by applying the
transformation to linear factors.  The expansion theorem states the
correctness of this transformation.

%% The correctness of the definition of linear factors is verified with an expansion theorem. It states that the set of
%% linear factors obtained from a formula is equivalent to the original formula.

\begin{theorem}[Expansion]\label{theorem:expansion}
  For all $\forma$, $\Theta (\LF (\forma)) \Leftrightarrow \forma$.
\end{theorem}

The partial derivative of a formula $\forma$ with respect to a symbol
$x\in\Sigma$ is a set of formulae $\FORMB$ such that  $x \sigma
\models \forma$ if and only if $\sigma \models \bigvee\FORMB$. 
Partial derivatives only need to be defined for formal conjunctions of
temporal formulae as we can apply the  $\SIMP$ operator first.

\begin{definition}[Partial Derivatives]
  The partial derivative of a formal conjunction of temporal formulae with respect
  to a symbol $\symx\in\Sigma$ is defined by
  \begin{align*}
    \pderiv{\forma}{\symx} &= \{ \forma' \mid \Angle{\monoa, \forma'} \in
                         \LF (\forma), \symx \models \monoa \}
                         & \text{ if $\forma$ is a temporal formula}\\
    \pderiv{\otop}{\symx} &= \{ \otop \} \\
    \pderiv{\forma\owedge\formb}{\symx}
                       &= \{ \forma' \owedge \formb' \mid \forma' \in
                         \pderiv\forma \symx, \formb' \in \pderiv\formb \symx
                       \}
                         \text.
  \end{align*}
\end{definition}

\begin{example}
  Continuing the example of $\alwaysL{\eventuallyL{p}}$, we find for $\symx\in\Sigma$:
  \begin{align*}
    \pderiv {\alwaysL {\eventuallyL{p}}} x
    &= \{ \forma' \mid \Angle{\monoa, \forma'} \in
                         \LF (\alwaysL {\eventuallyL{p}}), x \models \monoa \} \\
    &= \{ \forma' \mid \Angle{\monoa, \forma'} \in
                        \{ \Angle{p, \alwaysL{\eventuallyL{p}}}, \Angle{\true, \eventuallyL{p} \owedge
      \alwaysL{\eventuallyL{p}}} \}, x \models \monoa \} \\ 
    &=
      \begin{cases}
        \{ \alwaysL{\eventuallyL{p}},  \eventuallyL{p} \owedge
        \alwaysL{\eventuallyL{p}} \} & p \in \INT (x) \\
        \{ \eventuallyL{p} \owedge
        \alwaysL{\eventuallyL{p}}\} & p \notin \INT (x)
      \end{cases}
  \end{align*}
  As it is sufficient to define the derivative for temporal formulae, it only remains to explore the definition of
  $\pderiv{\eventuallyL p} x$.
  \begin{align*}
    \pderiv{\eventuallyL p} x
    &= \{ \forma' \mid \Angle{\monoa, \forma'} \in
      \LF (\eventuallyL p), x \models \monoa \} \\
    &= \{ \forma' \mid \Angle{\monoa, \forma'} \in
      \{ \Angle{p, \otop}, \Angle{\true, \eventuallyL{p}} \}, x \models \monoa \} \\
    &=
      \begin{cases}
        \{ \otop, \eventuallyL{p} \} & p \in \INT (x) \\
        \{ \eventuallyL{p} \} & p \notin \INT (x)
      \end{cases}
  \end{align*}
\end{example}

% Then
% \begin{align*}
%   \pderiv {\alwaysL \forma} x
%   &= \{ \forma' \mid \Angle{ x, \forma' } \in \LF (\alwaysL \forma)\} \\
%   &= \{ \forma' \owedge \alwaysL \forma \mid \Angle{ x, \forma' } \in \LF ( \forma)\}
% \end{align*}
% \begin{align*}
%   \pderiv {\alwaysL \forma} {xy}
%   &= \bigcup \{ \pderiv {\forma'} y \mid \Angle{ x, \forma' } \in \LF (\alwaysL \forma)\} \\
%   &= \bigcup \{ \pderiv{\forma' \owedge \alwaysL \forma} y \mid \Angle{ x, \forma' } \in \LF (
%     \forma)\}  \\
%   &= \bigcup \{ \pderiv{\forma'}y \owedge \pderiv{\alwaysL \forma} y \mid \Angle{ x, \forma' } \in \LF (
%     \forma)\}  \\
% \end{align*}

%-------------------------------------------------------------------------
\subsection{Properties of Partial Derivatives}

A descendant of a formula is either the formula itself or an element
of the partial derivative of a descendant by some symbol.
As in the regular expression case, the set of descendants of a fixed
LTL formula is finite. Our finiteness proof follows the method  suggested 
by Broda and coworkers~\cite{partial-derivative-plain-shuffle}.
We look at the set of iterated partial derivatives of a formula $\forma$,
which turns out to be just the set of temporal subformulae of $\forma$.
This set is finite and closed under the partial derivative operation. Thus, finiteness follows.

\begin{definition}[Iterated Partial Derivatives]
\bda{lcl}
\partial^+ (\Lit) &=&  \{ \Lit \}
\\
\partial^+ (\true) &=&  \{ \true \}
\\
\partial^+ (\false) &=&  \{\false \}
\\
\partial^+ (\forma \vee \formb) &=& \partial^+ (\forma)\cup \partial^+ (\formb)
\\
\partial^+ (\forma \wedge \formb) &=& \partial^+ (\forma) \cup \partial^+ (\formb) 
\\
\partial^+ (\nextL \forma) &=&  \{\nextL\forma\} \cup \partial^+ (\forma) 
\\
\partial^+ (\eventuallyL{\forma}) &=&  \{ \eventuallyL{\forma} \} \cup \partial^+ (\forma) 
\\
\partial^+ (\alwaysL{\forma}) &=& \{ \alwaysL \forma \} \cup \partial^+ (\forma) 
\\
\partial^+ (\untilL{\forma}{\formb}) &=& \{ \untilL{\forma}{\formb} \} \cup \partial^+ (\formb) \cup \partial^+ (\forma)
\\
\partial^+ (\releaseL{\forma}{\formb}) &=& \{ \releaseL{\forma}{\formb} \} \cup  \partial^+ (\formb) \cup \partial^+ (\forma)
\eda
\end{definition}

It is trivial to see that the set $\partial^+ (\forma)$ is finite
because it is a subset of the set of subformulae of $\forma$.
\begin{lemma}[Finiteness]
\label{le:finiteness}
  For all $\forma$, $\partial^+ (\forma)$ is finite.
\end{lemma}

The iterated partial derivative only consider subformulae
whereas the partial derivative elides disjunctions but  returns a set of formal conjunctions.
To connect both the following definition is required.

\begin{definition}[Subsets of Formal Conjunctions]
  For an ordered set $X = \{ x_1, x_2, \dots \}$, we define the set of all formal conjunctions of $X$
  as follows.
\begin{align*}
  \SET (X) &=
                  \{ \otop \} \cup \{ x_{i_1} \owedge \dots \owedge
             x_{i_n}  \mid n\ge 1, i_1 < i_2 < \dots < i_n \}
\end{align*}
We regard a subset of $\SET (X)$ as a positive Boolean formula over
$X$ in conjunctive normal form. 
\end{definition}
Clearly, if a set of formulae $\FORMA$ is
finite, then so is $\SET (\FORMA)$, where we assume an arbitrary, but
fixed total ordering on formulae.

The set of temporal subformulae of a given formula $\forma$ is also a formal conjunction of subformulae.
\begin{lemma}\label{lemma:simp-delta-plus}
  For all $\forma$, $\SIMP (\forma) \subseteq \SET (\partial^+ (\forma))$.
\end{lemma}

\begin{lemma}[Closedness under derivation] \label{le:closedness}~\\[-\baselineskip]
  \begin{enumerate}
  \item For all $x\in\Sigma$, $\pderiv \forma x \subseteq \SET( \partial^+ (\forma)) \}$. 
  \item For all $\forma' \in \partial^+ (\forma)$ and $x\in\Sigma$, $\pderiv {\forma'} x \subseteq \SET(\partial^+ (\forma))$.
  \end{enumerate}
\end{lemma}

From Lemmas~\ref{lemma:simp-delta-plus} and~\ref{le:closedness} it follows that
the set of descendants of a fixed LTL formula $\forma$ is finite.
In fact, we can show that the cardinality of this set is exponential
in the size of $\forma$.
We will state this result for
a more ``direct'' definition of partial derivatives
which does not require having to compute linear factors first.

\begin{definition}[Direct Partial Derivatives]
\label{def:direct-part-deriv}  
  Let $x \in \Sigma$. Then, $\partDeriv{\cdot}{x}$ maps LTL formulae to sets of LTL formulae and is defined as follows.
  \bda{lcl}
  \partDeriv\true x & = &  \{\true\} \\
  \partDeriv\false x & = & \{\} \\
  \partDeriv{\Lit}{x} & = & \left \{
                         \ba{ll}
                         \{ \true \} & x \models \Lit \\
                         \{ \}      & \mbox{otherwise}
                         \ea
                         \right.
  \\
  \partDeriv{\forma \vee \formb}{x} & = & \partDeriv{\forma}{x} \cup \partDeriv{\formb}{x}
  \\
  \partDeriv{\forma \wedge \formb}{x} & = & \{ \forma' \wedge \formb' \mid \forma' \in \partDeriv{\forma}{x},
                                                                          \formb' \in \partDeriv{\formb}{x} \}
  \\
  \partDeriv{\nextL{\forma}}{x} & = & \SIMP(\forma)
  \\
  \partDeriv{\untilL{\forma}{\formb}}{x} & = & \partDeriv{\formb}{x} \cup \{ \forma' \wedge \untilL{\forma}{\formb} \mid \forma' \in \partDeriv{\forma}{x} \}
  \\
  \partDeriv{\releaseL{\forma}{\formb}}{x} & = &
  \{ \forma' \wedge \formb' \mid \forma' \in \partDeriv{\forma}{x}, \formb' \in \partDeriv{\formb}{x} \}
  \cup \{ \formb' \wedge \releaseL{\forma}{\formb} \mid \formb' \in \partDeriv{\formb}{x} \}
  \\
  \partDeriv{\eventuallyL\forma}{x} & = & \partDeriv{\forma}{x} \cup \{ \eventuallyL\forma \}
  \\
  \partDeriv{\alwaysL\forma}{x} & = & \{ \forma' \wedge \alwaysL\forma \mid \forma' \in \partDeriv{\forma}{x} \}
  \eda
  where conjunctions of temporal formulae are normalized as usual.

  For $w \in \Sigma^*$, we define $\partDeriv{\forma}{\varepsilon} = \{ \forma \}$
  and $\partDeriv{\forma}{x w} = \bigcup_{\forma' \in \partDeriv{\forma}{x}} \partDeriv{\forma'}{w}$.
  For $L \subseteq \Sigma*$, we define $\partDeriv{\forma}{L} = \bigcup_{w \in L} \partDeriv{\forma}{w}$.
  We refer to the special case $\partDeriv{\forma}{\Sigma^*}$ as the set of \emph{partial derivative descendants}
  of $\forma$.
\end{definition}

\begin{example}
  Consider the formula $\alwaysL{\eventuallyL{p}}$.
  We calculate
  \bda{lcl}
  \partDeriv{\eventuallyL{p}}{p}
      & = & \{ \true, \eventuallyL{p} \}
  \\
  \partDeriv{\alwaysL\eventuallyL{p}}{p}
      & = & \{ \true \wedge \alwaysL\eventuallyL{p}, \eventuallyL{p} \wedge \alwaysL\eventuallyL{p} \}
  \\    &  & \mbox{(normalize)}
  \\    & = & \{ \alwaysL\eventuallyL{p}, \eventuallyL{p} \wedge \alwaysL\eventuallyL{p} \}
  \\
  \partDeriv{\eventuallyL{p} \wedge \alwaysL\eventuallyL{p}}{p}
  & = & \{ \true \wedge \true \wedge \alwaysL\eventuallyL{p},
           \eventuallyL{p} \wedge \alwaysL\eventuallyL{p},
          \true \wedge \eventuallyL{p} \wedge \alwaysL\eventuallyL{p},
          \eventuallyL{p} \wedge \eventuallyL{p} \wedge \alwaysL\eventuallyL{p} \}
 \\ & & \mbox{(normalize)}
 \\ & = & \{ \alwaysL\eventuallyL{p}, \eventuallyL{p} \wedge \alwaysL\eventuallyL{p} \}        
  \eda
\end{example}

\begin{lemma}
\label{le:part-deriv-direct}  
  For all $\forma$ and $x \in \Sigma$, $\pderiv{\forma}{x} = \partDeriv{\forma}{x}$.
      % \ms{add somewhere that in $\owedge$ we remove $\true$, so that
      % normalizations for $\owedge$ 
      %   and $\wedge$ are in 'sync'.}
      % \ms{we equate $\wedge = \owedge$}
\end{lemma}

The next result follows from Theorem \ref{theorem:expansion} and Lemma \ref{le:part-deriv-direct}.

\begin{lemma}
For all $\forma$, $\forma \Leftrightarrow \bigvee_{x \in \Sigma, \forma' \in \partDeriv{\forma}{x}} x \wedge \nextL{\forma'}$.
\end{lemma}

\begin{definition}
  The \emph{size} of a temporal formula $\forma$ is the sum of
  the number of literals, temporal and Boolean operators in $\forma$.  
\end{definition}

If $\forma$ has size $n$, the number of subformulae in $\forma$ is bounded by $O(n)$.

\begin{lemma}
\label{le:part-deriv-cardinality}
  For all $\forma$, the cardinality of
  $\partDeriv{\forma}{\Sigma^*}$ is bounded by $O(2^n)$ where
  $n$ is the size of $\forma$.
\end{lemma}

%-------------------------------------------------------------------------
%-------------------------------------------------------------------------  
\section{Alternating $\omega$-Automata}
\label{sec:construction}

We revisit the standard construction from LTL formula
to alternating $\omega$-automata as reported by
Vardi~\cite{Vardi:1994:NAA:645868.668514}. 
We observe that the definition of the transition function for formulae in PNF corresponds
to partial derivatives.

The transition function of an alternating automaton yields a set of
sets of states, which we understand as a disjunction of conjunctions
of states. The disjunction models the nondeterministic alternatives
that the automaton can take in a step, whereas the conjunction models
states that need to succeed together. Many presentations use positive
Boolean formulae at this point, our presentation uses the set of
minimal models of such formulae.
% \pt{we need to decide whether we use sets of states $\Power (Q)$ or
%   formal disjunctions of states $\SET (Q)$. Both would work, but the
%   definition of a run is simpler for $\Power (Q)$.}

%% MS: omit, doesn't seem to be necessary
%% \begin{definition}
%%   A \emph{positive boolean formula} over a set of atoms $X$ is defined by the grammar
%%   \begin{align*}
%%     \PBF (X) \ni \pbfa, \pbfb & ::= \otop \mid \obot \mid x \in X \mid \pbfa \ovee \pbfb \mid \pbfa \owedge \pbfb
%%   \end{align*}
%%   where $\otop$, $\obot$, $\ovee$, and $\owedge$ stand for truth, falsity, disjunction, and conjunction, respectively.
%% \end{definition}
%% We consider two positive boolean formulae equal if they can be interconverted by the laws of boolean algebra. That is,
%% $\ovee$ and $\owedge$ are associative, commutative, and idempotent; they distribute over one another; $\otop$ is
%% absorptive for $\ovee$ and unit for $\owedge$; whereas $\obot$ is absorptive for $\owedge$ and unit for $\ovee$. 
%% 
%% If $X$ is a set, define the set of all formal conjunctions
%% \begin{align*}
%%   \SET (X) &=
%%                   \{ \otop \} \cup \{ x_1 \owedge \dots \owedge x_n  \mid n\ge 1, x_i \in X \}
%% \end{align*}
%% where each conjunct appears at most once. We regard a subset of
%% $\SET (X)$ as a positive boolean formula in disjunctive normal form.

\begin{definition}
  A tuple $\A = (Q, \Sigma,  \delta, \pbfa_0, F)$ is
  an \emph{alternating $\omega$-automaton} (AA)
  \cite{DBLP:conf/ifipTCS/LodingT00} if $Q$ is a finite set of states,
  $\Sigma$  
  an alphabet, $\pbfa_0 \subseteq \Power (Q)$ a set of sets of states,
  $\delta : Q \times \Sigma \to \Power (Q)$ 
  a transition function, and $F\subseteq Q$ a set of 
  accepting states. 

  A \emph{run of $\A$ on a word $\sigma$} is a digraph $G= (V,E)$ with nodes $V \subseteq Q \times \nat$ and edges $E \subseteq
  \bigcup_{i\in\nat} V_i \times V_{i+1}$ where $V_i = Q \times \{i\}$, for all $i$.
  \begin{itemize}
  \item $\{ q \in Q \mid (q, 0) \in V\} \in \pbfa_0$.
  \item For all $i\in\nat$:
    \begin{itemize}
    \item If $(q', i+1) \in V_{i+1}$, then $((q,i),(q', i+1)) \in E$,
      for some $q\in Q$.
    \item If $(q,i) \in V_i$, then $\{ q' \in Q \mid ((q,i), (q', i+1)) \in E \} \in \delta (q, \sigma_i)$.
    \end{itemize}
  \end{itemize}

  A run $G$ on $\sigma$ is \emph{accepting} if every infinite path in
  $G$ visits a state in $F$ infinitely often (B{\"u}chi 
  acceptance). Define the language of $\A$ as
$$\L (\A) = \{ \sigma \mid \text{ there exists an accepting run of }\A\text{ on }\sigma \}.$$
\end{definition}

\begin{definition}[\cite{Vardi:1994:NAA:645868.668514,DBLP:conf/lics/MullerSS88}]
\label{def:vardi-aa}  
The alternating $\omega$-automaton $\A (\forma) = (Q, \Sigma, \delta,
\pbfa_0, F)$ resulting from $\forma$ is defined as follows. The set of
states is $Q = \partial^+ (\forma)$, the set of initial states
$\pbfa_0 = \SIMP(\forma)$, the set of accepting states $F = \{ \true
\} \cup \{ \releaseL{\forma}\formb \mid \releaseL{\forma}\formb \in Q
\}$, and the transition function $\delta$ is defined by induction on
the formula argument:
    \begin{itemize}
    \item $\delta (\true, x) = \{ \true \}$
    \item $\delta (\false, x) = \{ \}$
    \item $\delta(\Lit,x) = \{ \true \}$, \mbox{if $x \models \Lit$}
    \item $\delta(\Lit,x) = \{ \}$, \mbox{otherwise}
    \item $\delta(\forma \vee \formb, x) = \delta(\forma,x) \cup \delta(\formb,x) \}$
    \item $\delta(\forma \wedge \formb, x) = \{ q_1 \wedge q_2 \mid q_1 \in \delta(\forma,x), q_2 \in \delta(\formb,x) \}$
    \item $\delta(\nextL\forma, x) = \SIMP (\forma)$
    \item $\delta(\untilL\forma\formb,x) = \delta(\formb,x) \cup \{ q \wedge \untilL\forma\formb \mid q \in \delta(\forma,x) \}$
    \item   $\delta(\releaseL\forma\formb,x)
  = \{ q_1 \wedge q_2 \mid q_1 \in \delta(\forma,x), q_2 \in \delta(\formb,x) \}
  \cup \{ q \wedge \releaseL\forma\formb \mid q \in \delta(\formb,x) \}$
%%    \item $\delta(\releaseL\forma\formb,x) = \delta(\formb,x) \cup \{ q \wedge \releaseL\forma\formb \mid q \in \delta(\forma,x) \}$      
   \end{itemize}
\end{definition}
We deviate slightly from Vardi's original definition by representing disjunction
as a set of states. For example, in his definition $\delta (\false, x)
= \false$, which is equivalent to the empty disjunction. Another
difference is that we only consider formulae in PNF whereas Vardi covers LTL in general.
Hence, Vardi's formulation treats negation by extending the set of
states with negated subformulae. For example, we find
$\delta(\neg\forma,x) = \overline{\delta(\forma,x)}$ 
where $\overline{\FORMA}$ calculates the dual of a set $\FORMA$ of formulae
obtained by application of the de Morgan laws.
The case for negation can be dropped because we assume that formulae are in PNF.
In exchange, we need to state the cases for $\forma\vee\formb$ and for
$\releaseL\forma\formb$ 
% $\delta(\releaseL\forma\formb,x)
%   = \{ q_1 \wedge q_2 \mid q_1 \in \delta(\forma,x), q_2 \in \delta(\formb,x) \}
%   \cup \{ q \wedge \releaseL\forma\formb \mid q \in \delta(\formb,x) \}$
%
which can be derived easily from Vardi's formulation by
exploiting standard LTL equivalences.
% such as $\neg \neg \forma = \forma$,
% $\releaseL\forma\formb = \untilL{\neg\forma}{\neg\formb}$.

% \ms{correctness of the construction derived from
% \cite{Vardi:1994:NAA:645868.668514}, actual proof rather short!
% reference to ``Weak alternating automata give a simple explanation
% of why most temporal and dynamic logics are decidable in exponential
% time'' for which i couldn't find an online copy} 

The accepting states in Vardi's construction are all subformulae of
the form $\neg(\untilL\forma\formb)$, but $\neg(\untilL\forma\formb) =
\releaseL{(\neg\forma)}{(\neg\formb)}$, which matches our
definition and others in the literature \cite{DBLP:journals/fmsd/FinkbeinerS04}.

Furthermore, our construction adds $\true$ to the set of accepting
states, which is not present in Vardi's paper. It turns out that
$\true$ can be eliminated from the accepting states if we set
$\delta (\true, x) = \{\}$. This change transforms
an infinite path with infinitely many $\true$ states into a finite
path that terminates when truth is established. Thus, it does not
affect acceptance of the AA.

The same definition is given by Pel\'{a}nek and Strej\v{c}ek
\cite{DBLP:conf/wia/PelanekS05} who note that the resulting automaton
is restricted to be a 1-weak alternating automaton. For this class of
automata there is a translation back to LTL.

We observe that
the definition of the transition function in Definition~\ref{def:vardi-aa}
corresponds to the direct definition of partial derivatives
in Definition~\ref{def:direct-part-deriv}.

\begin{lemma}
  Let $\A (\forma)$ be the alternating $\omega$-automaton for a formula $\forma$
  according to Definition~\ref{def:vardi-aa}.
  For each $\formb\in Q$ and $x\in\Sigma$, we have that $\delta (\formb, x) = \partDeriv{\formb}{x}$.
\end{lemma}

Finally we provide an independent correctness result for the
translation from LTL to AA that relies on the correctness of our
construction of linear factors.

\begin{theorem}
  \label{th:ltl-to-aa}
  Let $\forma$ be an LTL formula.
  Consider the alternating automaton $\A (\forma)$ given by
  \begin{itemize}
  \item $Q = \partial^+ (\forma)$,
  \item $\delta (\formb, x) = \pderiv{\formb}{x}$, for all $\formb\in Q$ and $x\in\Sigma$,
  \item $\pbfa_0 = \SIMP(\forma)$,
  \item $F = \{ \true \} \cup \{ \releaseL{\forma}\formb \mid \releaseL{\forma}\formb \in Q \}$.
  \end{itemize}
  Then, $\L (\forma) = \L (\A (\forma))$ using  the B\"uchi acceptance condition.  
\end{theorem}

%-------------------------------------------------------------------------
%-------------------------------------------------------------------------  
\section{Semantic Tableau}
\label{sec:semantic-tableau}

We revisit Wolper's \cite{wolper85-ltl-tableau} semantic tableau method for LTL
to check satisfiability of a formula $\forma$.
A tableau is represented as a directed graph built from nodes where nodes denote sets of formulae.
A tableau starts with the initial node $\{ \forma \}$ and new nodes are generated by decomposition of existing nodes, i.e.~formulae.
Wolper's method requires a post-processing phase where unsatisfiable nodes are eliminated.
The formula $\forma$ is satisfiable if there is a satisfiable path in the tableau.
We observe that decomposition can be explained in terms of linear factors
and some of the elimination (post-processing) steps can be obtained for free.

We largely follow Wolper's notation but start from formulae in PNF.
In the construction of a tableau,
a formula $\forma$ may be \emph{marked}, written as $\markL{\forma}$.
A formula is \emph{elementary} if
it is a literal or its outermost connective is $\nextL{}$.
We write $S$ to denote a set of formulae. Hence, each node is represented by some $S$.
A node is called a \emph{state} if the node consists solely of elementary or marked formulae.
A node is called a \emph{pre-state} if it is the initial node or the
immediate child of a state.

\begin{definition}[Wolper's Tableau Decision Method \cite{wolper85-ltl-tableau}]
  Tableau construction for $\forma$ starts with node $S = \{ \forma \}$.
  New nodes are created as follows.
    \begin{itemize}
    \item Decomposition rules: For each non-elementary unmarked $\forma \in S$
      with decomposition rule $\forma \rightarrow \{
      S_1,\dots,S_k \}$ as defined below, 
      create $k$ child nodes where the $i$th child is of the form $(S - \{ \forma \}) \cup S_i \cup \{\markL{\forma}\}$.

          \bda{crcl}
    \rlabel{D1} & \forma \vee \formb & \rightarrow & \{ \{\forma\}, \{\formb \} \}
    \\
    \rlabel{D2} & \forma \wedge \formb & \rightarrow & \{ \{\forma, \formb \} \}    
    \\
    \rlabel{D3} & \eventuallyL{\forma} & \rightarrow & \{ \{\forma \}, \{\nextL{\eventuallyL{\forma}} \} \}
    \\
    \rlabel{D4} & \alwaysL{\forma} & \rightarrow & \{ \{ \forma, \nextL{\alwaysL{\forma}} \} \}
    \\
    \rlabel{D5} & \untilL{\forma}{\formb} & \rightarrow & \{ \{ \formb \}, \{ \forma, \nextL{(\untilL{\forma}{\formb})} \} \}
    \\
    \rlabel{D6} & \releaseL{\forma}{\formb} & \rightarrow & \{ \{ \formb, \forma \vee \nextL{(\releaseL{\forma}{\formb})} \} \}
    \eda
%     \ms{seems typo in Wolper's paper, misses $\{ \dots \}$}
%     \ms{In the LF notation, rhs written as
%       $\{ \{\forma, \formb \}, \{ \formb, \nextL{(\releaseL{\forma}{\formb})} \} \}$.}
%     \pt{yep, that looks right}      
    \item Step rule: For each node $S$ consisting of only elementary or marked formulae,
          create a child node $\{ \forma \mid \nextL{\forma} \in S
          \}$. Just create an edge if the node already exists.
    \end{itemize}

  Elimination of (unsatisfiable) nodes proceeds as follows.
  A node is eliminated if one of the conditions E1-3 applies.   
    \begin{itemize}
        \item E1: The node contains $p$ and its negation.
        \item E2: All successors have been eliminated.
        \item E3: The node is a pre-state and contains a formula of the form $\eventuallyL{\formb}$
          or $\untilL{\forma}{\formb}$ that is not satisfiable.
    \end{itemize}
  A formula $\eventuallyL{\formb}$, $\untilL{\forma}{\formb}$ is satisfiable in a pre-state,
    if there is a path in the tableau leading from that pre-state to a node containing
    the formula $\formb$.
%% MS: omit, I believe this has been adressed, check Wolper's correctness proof.    
%% \pt{I wonder if we also need to state a criterion about $\alwaysL\forma$.}
\end{definition}

\begin{theorem}[Wolper \cite{wolper85-ltl-tableau}]
  An LTL formula $\forma$ is satisfiable iff the initial node generated
  by the tableau decision procedure is not eliminated.
\end{theorem}

\begin{figure}[tp]
  {\tiny
    \bda{c}
  \xymatrix{ & S_0 = \{ \alwaysL{p} \wedge \eventuallyL{\neg p} \} \ar[d]^2 &&
    \\     & S_1 = \{ \markL{\alwaysL{p} \wedge \eventuallyL{\neg p}}, \alwaysL{p}, \eventuallyL{\neg p} \} \ar[d]^4
    \\     & S_2 = \{ \markL{\alwaysL{p} \wedge \eventuallyL{\neg p}}, \markL{\alwaysL{p}}, \eventuallyL{\neg p}, p, \nextL{\alwaysL{p}} \} \ar[dl]^3 \ar[d]^3
    \\   S_3 = \{ \markL{\alwaysL{p} \wedge \eventuallyL{\neg p}}, \markL{\alwaysL{p}}, \markL{\eventuallyL{\neg p}}, p, \nextL{\alwaysL{p}}, \neg p \} & S_4 = \{ \markL{\alwaysL{p} \wedge \eventuallyL{\neg p}}, \markL{\alwaysL{p}}, \markL{\eventuallyL{\neg p}}, p, \nextL{\alwaysL{p}}, \nextL{\eventuallyL{\neg p}} \} \ar[d]
    \\      & S_5 = \{ \alwaysL{p}, \eventuallyL{\neg p} \} \ar[d]^4
    \\      & S_6 = \{ \markL{\alwaysL{p}}, \eventuallyL{\neg p}, p, \nextL{\alwaysL{p}} \} \ar[dl]^3 \ar[d]^3
    \\  S_7 = \{ \markL{\alwaysL{p}}, \markL{\eventuallyL{\neg p}}, p, \nextL{\alwaysL{p}}, \neg p \} & S_8 = \{ \markL{\alwaysL{p}}, \markL{\eventuallyL{\neg p}}, p, \nextL{\alwaysL{p}}, \nextL{\eventuallyL{\neg p}} \} \ar@/_8pc/[uu] }
  \eda
  }
  \caption{Tableau before elimination: $\alwaysL{p} \wedge \eventuallyL{\neg p}$}
  \label{fig:tableau}
\end{figure}

\begin{example}
  Consider $\alwaysL{p} \wedge \eventuallyL{\neg p}$.
  Figure~\ref{fig:tableau} shows the tableau generated before elimination.
  In case of decomposition, edges are annotated with the number of the
  respective decomposition rule. 
  For example, from the initial node $S_0$ we reach node $S_1$ by decomposition via \rlabel{D2}.
  Node $S_4$ consists of only elementary and marked nodes and therefore we apply the step rule
  to reach node $S_5$.
  The same applies to node $S_3$. For brevity, we ignore its child node because this node
  is obviously unsatisfiable (E1). The same applies to node $S_7$.

  We consider elimination of nodes. Nodes $S_3$, $S_4$, $S_7$ and $S_8$ are states.
  Therefore, $S_0$ and $S_5$ are pre-states.
    Nodes $S_3$ and $S_7$ can be immediately eliminated due to E1.
  Node $S_5$ contains $\eventuallyL{\neg p}$. This formula is not satisfiable because there is not path from $S_5$
  along which we reach a node which contains $\neg p$. Hence, we eliminate $S_5$ due to E3.
  All other nodes are eliminated due to E3.
  Hence, we conclude that the formula $\alwaysL{p} \wedge \eventuallyL{\neg p}$ is unsatisfiable.
\end{example}  

We argue that marked formulae and intermediate nodes are not essential in Wolper's tableau construction.
Marked formulae can simply be dropped and intermediate nodes can be removed
by exhaustive application of decomposition.
This optimization reduces the size of the tableau
and allows us to establish a direct connection between states/pre-states and linear factors/partial derivatives.

\begin{definition}[Decomposition and Elimination via Rewriting]
 \label{def:decomp-elim}
%We write $N$ to denote a set of nodes $S$.
We define a rewrite relation among sets of sets of nodes
ranged over by $N$.
\bda{c}
\myirule{\text{``}\forma \rightarrow \{S_1,\dots,S_n\}\text{''} \in \{D1, \dots, D6\}}
        {\{ S \cup \{\forma\} \} \cup N \rightarrowtail
          \{ S \cup S_1 \} \cup \dots \cup \{ S \cup S_n \} \cup N}
\\
\\
\myirule{ % N = \{ S_1, \dots, S_n \} \ \ \ \
         N' = \{ S \mid S \in N \wedge (\forall \Lit\in S)~ \neg\Lit \notin S \}}
        {N \rightarrowtail N'}
\eda
where the premise of the first rule corresponds to one of the
decomposition rules D1-D6. The second rule corresponds to the
elimination rule E1 applied globally.
We write $N_1 \rightarrowtail^* N_k$ for $N_1 \rightarrowtail \dots \rightarrowtail N_k$
where no further rewritings are possible on $N_k$.
We write $\forma \rightarrowtail^* N$ as a shorthand for $\{ \{ \forma \} \} \rightarrowtail^* N$.
\end{definition}

As the construction does not mark formulae, we call $S$ a state
node if $S$ only consists of elementary formulae.
By construction, for any set of formulae $S$ we find that $\{S\} \rightarrowtail^* N$
for some $N$ which only consists of state nodes.
In our optimized Wolper-style tableau construction,
each $S' \in N$ is a `direct' child of $S$ where intermediate nodes are skipped.
We also integrate the elimination condition E1 into the construction of new nodes.

The step rule is pretty much the same as in Wolper's formulation.
The (mostly notational) difference is that we represent a pre-state node with a single formula.
That is, from state node $S$ we generate the child pre-state node $\{ \bigwedge_{\nextL\formb \in S} \formb \}$
whereas Wolper generates $\{ \formb \mid \nextL\formb \in S \}$.

\begin{definition}[Optimized Tableau Construction Method]
  We consider tableau construction for $\forma$.
  We assume that $Q$ denotes the set of pre-state formulae generated so far
  and $Q_j$ denotes the set of active pre-state nodes in the $j$-th construction step.
  We start with $Q = Q_0 = \{ \forma \}$.
  We consider the $j$-th construction step.
    \begin{description}
    \item[Decomposition:]
      For each node $\formb \in Q_j$ we build $\{ \{ \formb \} \} \rightarrowtail^* \{ S_1, \dots, S_n \}$
      where each state node $S_i$ is a child of pre-state node $\{ \formb \}$.
    \item[Step:] For each state node $S_i$, we build $\forma_i = \bigwedge_{\nextL\forma \in S_i} \formb$ where
                 pre-state node $\{ \forma_i \}$ is a child of $S_i$.
                 We set $Q_{j+1} = \{ \forma_1, \dots, \forma_n \} - Q$
                 and then update the set of pre-state formulae generated so far by
                 setting $Q = Q \cup \{ \forma_1, \dots, \forma_n \}$. 
    \end{description}
    Construction continues until no new children are created.
\end{definition}

\begin{theorem}[Correctness Optimized Tableau]
\label{th:correctness-optimized-variant}  
  For all $\forma$, $\forma$ is satisfiable iff the initial node generated
  by the optimized Wolper-style tableau decision procedure is not eliminated
  by conditions E2 and E3.
\end{theorem}

\begin{example}
 \label{ex:opt1ex}  
  Consider $\alwaysL{p} \wedge \eventuallyL{\neg p}$.
  Our variant of Wolper's tableau construction method yields the following.
  \bda{c}
  \xymatrix{S_0 = \{ \alwaysL{p} \wedge \eventuallyL{\neg p} \} \ar[d]^{decomp}
        \\ S_4' = \{ p, \nextL\alwaysL{p}, \nextL\eventuallyL{\neg p} \} \ar@/_8pc/[u]_{step}}
  \eda
  Node $S_4'$ corresponds to node $S_4$ in Figure~\ref{fig:tableau}.
  Nodes $S_1$, $S_2$, and $S_3$ from the original construction do not
  arise in our variant because we skip intermediate nodes 
  and eliminate aggressively during construction whereas Wolper's
  construction method gives rise $S_5$. 
  We avoid such intermediate nodes and immediately link $S_4'$ to the initial node $S_0$. 
\end{example}

Next, we show that states in the optimized representation of Wolper's tableau
method correspond to linear factors and pre-states correspond to partial derivatives.

Let $S = \{ \Lit_1,\dots,\Lit_n, \nextL{\forma_1},\dots,\nextL{\forma_m} \}$
be a (state) node.
We define $\sem{S} = \Angle{\Lit_1 \scm \dots \scm \Lit_n, \forma_1 \owedge \dots \owedge \forma_m}$
where for cases $n=0$ and $m=0$ we assume $\true$.
Let $N = \{ S_1, \dots, S_n \}$ where each $S_i$ is a state.
We define $\sem{N} = \{\sem{S_1}, \dots, \sem{S_n} \}$.

\begin{lemma}
\label{le:state-vs-lf}
Let $\forma \not= \false$, $N$ such that $\forma \rightarrowtail^* N$.
Then, we find that $\LF (\forma) = \sem{N}$.
\end{lemma}
Case $\false$ is excluded due to $LF (\false) = \{ \}$.

So, any state node generated during the optimized Wolper tableau construction
corresponds to an element of a linear factor.
An immediate consequence is that each pre-state corresponds to a partial derivative.
Hence, we can reformulate the optimized Wolper tableau construction as follows.

\begin{theorem}[Tableau Construction via Linear Factors]
 \label{th:tableau-lf}  
  The optimized variant of Wolper's tableau construction for $\forma$ can be
  obtained as follows.
  \begin{enumerate}\item 
    Each formula $\formb\not= \true$ in the set of all partial
    derivative descendants $\partDeriv{\forma}{\Sigma^*}$ corresponds
    to a pre-state.
  \item For each $\formb \in
    \partDeriv{\forma}{\Sigma^*}$ where $\formb \not= \true$, each
    $\Angle{\monob,\formb'} \in \LF (\formb)$ is state where
    $\Angle{\monob,\formb'}$ is a child of $\formb$, and if $\formb' \not=
    \true$, $\formb'$ is a child of $\Angle{\monob,\formb'}$.

  \end{enumerate}\end{theorem}
We exclude $\true$ because Wolper's tableau construction stops once we reach $\true$.

%% MS: reformulation of Example~\ref{ex:opt1ex}, omitted, not much space left
%% \begin{example}
%%   Consider $\alwaysL{p} \wedge \eventuallyL{\neg p}$.
%%   \bda{lcl}
%%   \LF (\alwaysL{p}) & = & \{ \Angle{p,\alwaysL{p}} \}
%%   \\
%%   \LF (\eventuallyL{\neg p}) & = & \{ \Angle{\neg p, \true}, \Angle{\true, \eventuallyL{\neg p}} \}
%%   \\
%%   \LF (\alwaysL{p} \wedge \eventuallyL{\neg p}) & = & \{ \Angle{p,\alwaysL{p} \wedge \eventuallyL{\neg p}} \}
%%   \eda
%%   %
%%   Here's the LF/PD-based tableau.
%%   \bda{c}
%%   \xymatrix{ \alwaysL{p} \wedge \eventuallyL{\neg p} \ar[d]^{LF} 
%%     \\      \Angle{p,\alwaysL{p} \wedge \eventuallyL{\neg p}} \ar@/_8pc/[u]_{PD}}
%%   \eda
%% \end{example}

\begin{example}
  Consider $\neg p \wedge \nextL{\neg p} \wedge \untilL{q}{p}$ where
  \bda{lcl}
  \LF (\neg p) & = & \{ \Angle{\neg p, \true} \}
  \\
  \LF (\true) & = & \{ \Angle{\true, \true} \}
  \\
  \LF (\nextL{\neg p}) & = & \{ \Angle{\true, \neg p} \}
  \\
  \LF (\untilL{q}{p}) & = & \{ \Angle{p, \true}, \Angle{q, \untilL{q}{p}} \}
  \\
  \LF (\neg p \wedge \untilL{q}{p}) & = & \{ \Angle{\neg p \wedge q, \untilL{q}{p}} \}
  \\
  \LF (\neg p \wedge \nextL{\neg p} \wedge \untilL{q}{p}) & = & \{ \Angle{\neg p \wedge q, \neg p \wedge \untilL{q}{p}} \}
  \eda
  We carry out the tableau construction using linear factors notation
  where we use LF to label pre-state (derivatives) to state (linear factor) relations
  and PD to label state to pre-state relations.
  \bda{c}
  
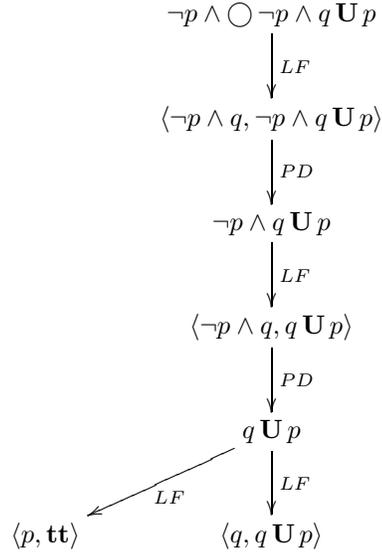
\xymatrix{ & \neg p \wedge \nextL{\neg p} \wedge \untilL{q}{p} \ar[d]^{LF} &&
    \\  & \Angle{\neg p \wedge q, \neg p \wedge \untilL{q}{p}} \ar[d]^{PD}
    \\  & \neg p \wedge \untilL{q}{p}  \ar[d]^{LF}
    \\ \ & \Angle{\neg p \wedge q, \untilL{q}{p}} \ar[d]^{PD}
    \\ & \untilL{q}{p} \ar[dl]^{LF} \ar[d]^{LF} 
    \\ \Angle{p,\true}
       %% \ar[d]^{PD}
    & \Angle{q,\untilL{q}{p}}
     %% \ar@/_2pc/[u]_{PD}
%%     \\ \true \ar[d]^{LF}
    %%     \\ \Angle{\true, \true} \ar@/^2pc/[u]}_{PD}
    }
  \eda
\end{example}

The reformulation of Wolper's tableau construction in terms of linear factors
and partial derivatives allows us to establish a close connection to the construction
of Vardi's alternating $\omega$-automaton.
Each path in the tableau labeled by LF and PD corresponds to a transition step in the automaton.
The same applies to transitions with one exception.
In Wolper's tableau, the state $\Angle{\Lit, \true}$ is considered final whereas
in Vardi's automaton we find transitions $\delta(\Lit, \true) = \{ \true \}$.
So, from Theorem~\ref{th:ltl-to-aa} and Theorem~\ref{th:tableau-lf}
we can derive the following result.

\begin{corollary}
  Vardi's alternating $\omega$-automaton derived from an LTL formula
  is isomorphic to Wolper's optimized LTL tableau construction
  assuming we ignore transitions $\delta(\Lit, \true) = \{ \true \}$.
\end{corollary}

%% MS: todo
%% %-------------------------------------------------------------------------
%% %-------------------------------------------------------------------------
%% \section{Related Work and Conclusion}
%% 
%% \ms{We have about 1 page left.
%%     (1) Hightlights of proof?
%%     (2) Example for AA construction?
%%     (3) Related works? Can keep this short. Mostly leads to future work, see next.
%%     (4) Discuss Reynold's tree-based tableau construction?
%%         I do wonder if we get this (almost) for free if we use Haskell's laziness, that is,
%%         run construction and search concurrently.
%%     (5) Surprisingly, I couldn't find a general emptiness check for AA (apparantly tricky).
%%     (6) Comments can be removed i believe.
%%     (7) Can create TR with full appendix.
%%     }

\bibliographystyle{plain}
\bibliography{main}

\pagebreak

%-------------------------------------------------------------------------
%-------------------------------------------------------------------------
\section{Proofs}

%-------------------------------------------------------------------------
\subsection{Proof of Theorem \ref{theorem:expansion}}

\begin{proof}
  Show by induction on $\forma$:
  for all $\sigma\in\Sigma^\omega$, $\sigma \models \forma$ iff $\sigma \models \Theta (\LF (\forma))$.

  \textbf{Case }$p$.
  \begin{align*}
    \Theta(\LF (p))
    & = \Theta(\{ \Angle{p, \otop} \})  = p \wedge \nextL\true  \Leftrightarrow p 
  \end{align*}

  \textbf{Case }$\neg p$. Analogous.

  \textbf{Case }$\true$.
  \begin{align*}
    \Theta (\LF (\true))
    &= \Theta (\{ \Angle{\true, \otop} \}) = \true \wedge \nextL\true \Leftrightarrow \true
  \end{align*}

  \textbf{Case }$\false$.
  \begin{align*}
    \Theta(\LF (\false))
    &= \Theta (\{ \}) = \false
  \end{align*}

  \textbf{Case }$\forma \vee \formb$.
  \begin{align*}
    \Theta (\LF (\forma \vee \formb))
    &= \Theta (\LF (\forma) \cup \LF (\formb)) = \Theta (\LF (\forma)) \vee \Theta( \LF (\formb))
  \end{align*}
  Now
  \begin{align*}
    \sigma \models \forma \vee \formb
    & \Leftrightarrow     (\sigma \models \forma) \vee (\sigma \models \formb) \\
    & \text{by IH}\\
    & \Leftrightarrow (\sigma \models \Theta (\LF (\forma))) \vee  (\sigma \models \Theta (\LF (\formb))) \\
    & \Leftrightarrow (\sigma \models \Theta (\LF (\forma)) \vee \Theta (\LF (\formb)))
  \end{align*}

  \textbf{Case }$\forma \wedge \formb$.
  \begin{align*}
    \Theta (\LF (\forma \wedge \formb))
    &= \Theta (\{ \Angle{\monoa\scm\monob,\forma' \owedge \formb'} \mid \Angle{\monoa,\forma'} \in \LF (\forma), \Angle{\monob,\formb'} \in \LF (\formb) \}) \\
    &= \bigvee \{ (\monoa\scm\monob) \wedge \nextL(\forma' \owedge \formb') \mid \Angle{\monoa,\forma'} \in \LF (\forma), \Angle{\monob,\formb'}
      \in \LF (\formb) \} 
  \end{align*}
  Now
  \begin{align*}
    \sigma &\models \forma \wedge \formb \\
    &\Leftrightarrow (\sigma \models \forma) \wedge (\sigma \models \formb)\\
    & \text{by IH} \\
    &\Leftrightarrow (\sigma \models \Theta(\LF(\forma))) \wedge (\sigma \models \Theta(\LF(\formb))) \\
    &\Leftrightarrow
      (\sigma \models \bigvee \{ \monoa \wedge \nextL{\forma'} \mid \Angle{\monoa, \forma'} \in \LF (\forma)\})
      % \\&\qquad
      \wedge
        (\sigma \models \bigvee \{ \monob \wedge \nextL{\formb'} \mid \Angle{\monob, \formb'} \in \LF (\formb)\})
    \\
    &\Leftrightarrow
      \sigma \models (\bigvee \{ \monoa \wedge \nextL{\forma'} \mid \Angle{\monoa, \forma'} \in \LF (\forma)\})
      \wedge (\bigvee \{ \monob \wedge \nextL{\formb'} \mid \Angle{\monob, \formb'} \in \LF (\formb)\})
    \\
    &\Leftrightarrow
      \sigma \models (\bigvee \{ \monoa \wedge \nextL{\forma'} \wedge \monob \wedge \nextL{\formb'} \mid
      \Angle{\monoa, \forma'} \in \LF (\forma), \Angle{\monob, \formb'} \in \LF (\formb)\})
      \intertext{
      by Lemma~\ref{lemma:conjunction-monomials} $\monoa \wedge \monob \Leftrightarrow \Theta (\monoa \scm \monob)$
      }
    &\Leftrightarrow
      \sigma \models (\bigvee \{ (\monoa \scm \monob) \wedge \nextL{\forma'} \wedge \nextL{\formb'} \mid
      \Angle{\monoa, \forma'} \in \LF (\forma), \Angle{\monob, \formb'} \in \LF
      (\formb)\})
    \\
    &\Leftrightarrow
      \sigma \models (\bigvee \{ (\monoa \scm \monob) \wedge \nextL{(\forma' \owedge \formb')} \mid
      \Angle{\monoa, \forma'} \in \LF (\forma), \Angle{\monob, \formb'} \in \LF (\formb) \})
  \end{align*}

  \textbf{Case }$\nextL{\forma}$. (using Lemma~\ref{lemma:simp-correct})
  \begin{align*}
    \Theta(\LF(\nextL{\forma}))
    &= \Theta (\{ \Angle{\true,\forma'} \mid \forma' \in \SIMP (\forma) \}) \\
    &=\bigvee \{ \true \wedge \nextL{\forma'} \mid \forma' \in \SIMP (\forma) \}\\
    &= \nextL{(\bigvee \SIMP (\forma))} \\
    &\Leftrightarrow\nextL{\forma}
  \end{align*}

  % \textbf{Case }$\eventuallyL \forma$.
  % \begin{align*}
  %   \Theta (\LF (\eventuallyL \forma))
  %   &= \Theta (\LF (\forma) \cup \{ \Angle{\true,\eventuallyL{\forma}} \}) \\
  %   &\Leftrightarrow \Theta (\LF (\forma)) \vee \true \wedge \nextL{\eventuallyL{\forma}} \\
  %   &\Leftrightarrow \Theta (\LF (\forma)) \vee \nextL{\eventuallyL{\forma}}\\
  %   & \text{by IH} \\
  %   &\Leftrightarrow \forma \vee  \nextL{\eventuallyL{\forma}} \\
  %   &\Leftrightarrow \eventuallyL{\forma}
  % \end{align*}
  
  % \textbf{Case }$\alwaysL \forma$.
  % \begin{align*}
  %   \Theta (\LF (\alwaysL \forma))
  %   &= \Theta (\{ \Angle{\monoa,\forma' \owedge \alwaysL{\forma}} \mid \Angle{\monoa,\forma'} \in
  %     \LF (\forma) \}) \\
  %   &= \bigvee \{ \monoa \wedge \nextL{(\forma' \owedge \alwaysL{\forma})} \mid
  %     \Angle{\monoa,\forma'} \in      \LF (\forma) \} \\
  %   & \Leftrightarrow (\bigvee \{ \monoa \wedge \nextL{\forma'} \mid \Angle{\monoa,\forma'}
  %     \in \LF (\forma) \}) \wedge \nextL{\alwaysL \forma} \\
  %   & \Leftrightarrow \Theta (\LF (\forma)) \wedge \nextL{\alwaysL \forma} \\
  %   & \text{by IH} \\
  %   & \Leftrightarrow \forma \wedge \nextL{\alwaysL \forma} \\
  %   & \Leftrightarrow \alwaysL \forma
  % \end{align*}

  \textbf{Case }$\untilL{\forma}{\formb}$. 
  \begin{align*}
    \Theta (\untilL{\forma}{\formb})
    &= \Theta (\LF (\formb) \cup \{ \Angle{\monoa,\forma' \owedge \untilL{\forma}{\formb}}
      \mid \Angle{\monoa,\forma'} \in \LF (\forma) \}) \\
    &= \Theta (\LF (\formb)) \vee \bigvee \{ {\monoa \wedge \nextL{(\forma' \owedge \untilL{\forma}{\formb})}}
      \mid \Angle{\monoa,\forma'} \in \LF (\forma) \} \\
    &\Leftrightarrow \Theta (\LF(\formb)) \vee \bigvee \{ \monoa \wedge
      \nextL{\forma'} \mid \Angle{\monoa, \forma'} \in \LF(\forma)\} \wedge \nextL{(\untilL{\forma}{\formb})} \\
    &\Leftrightarrow \Theta (\LF(\formb)) \vee (\Theta(\LF(\forma)) \wedge \nextL{(\untilL{\forma}{\formb})}) \\
    & \text{by IH} \\
    &\Leftrightarrow \formb \vee (\forma \wedge \nextL{(\untilL{\forma}{\formb})}) \\
    &\Leftrightarrow \untilL{\forma}{\formb}
  \end{align*}

  \textbf{Case }$\releaseL{\forma}{\formb}$.
  \begin{align*}
    \Theta ( \LF (\releaseL{\forma}{\formb}))
    & = \Theta(
    \ba{l}
      \{  \Angle{\monoa \scm \monob, \forma' \owedge \formb'} \mid
           \Angle{\monoa, \forma'} \in \LF (\forma),
           \Angle{\monob, \formb'} \in \LF (\formb)
        \} \\ \cup \{
          \Angle{\monob, \formb' \owedge \releaseL{\forma}{\formb}} \mid
          \Angle{\monob,\formb'} \in \LF (\formb)\}
      \ea          
      ) \\
      & = \ba{l} \bigvee_{\Angle{\monoa, \forma'} \in \LF (\forma), \Angle{\monob, \formb'} \in \LF (\formb)}
                   (\Theta (\monoa \scm \monob) \wedge \nextL{(\forma' \owedge \formb')})
           \\ \vee  \bigvee_{\Angle{\monob,\formb'} \in \LF (\formb)}
                  (\Theta( \monob) \wedge \nextL{ (\formb' \wedge \releaseL{\forma}{\formb})})
           \ea
           \\
           & \text{by Lemma~\ref{lemma:conjunction-monomials} and the fact that $\nextL{(\forma \wedge \formb)} \Leftrightarrow \nextL{\forma} \wedge \nextL{\formb}$} \\
      &\Leftrightarrow
         \ba{l} \bigvee_{\Angle{\monoa, \forma'} \in \LF (\forma), \Angle{\monob, \formb'} \in \LF (\formb)}
                   (\Theta (\monoa) \wedge \Theta (\monob) \wedge \nextL{\forma'}  \wedge \nextL{\formb'})
           \\ \vee  \bigvee_{\Angle{\monob,\formb'} \in \LF (\formb)}
               (\Theta( \monob) \wedge \nextL{\formb'} \wedge \nextL{(\releaseL{\forma}{\formb})})
           \ea
           \\
           &\text{by repeated application of the following distributivity laws} \\
           & (\forma_1 \wedge \forma_2) \vee (\forma_1 \wedge \forma_3) \Leftrightarrow \forma_1 \wedge (\forma_2 \vee \forma_3) \\
           & (\forma_1 \wedge \forma_2) \vee (\forma_3 \wedge \forma_2) \Leftrightarrow (\forma_1 \vee \forma_3) \wedge \forma_2 \\
           &\Leftrightarrow
           \ba{l} \bigvee_{\Angle{\monob,\formb'} \in \LF (\formb)}
                     (\Theta (\monob) \wedge \nextL{\formb'})
                     \\ \wedge (((\bigvee_{\Angle{\monoa,\forma'} \in \LF (\forma)}
                             (\Theta (\monoa) \wedge \nextL{\forma'}))) \vee \nextL{(\releaseL{\forma}{\formb})})
           \ea
           \\
           &= \Theta(\LF(\formb)) \wedge  (\Theta(\LF(\forma)) \vee \nextL{(\releaseL{\forma}{\formb})}) \\
           &\text{by IH} \\
           & \Leftrightarrow \formb \wedge (\forma \vee \nextL{(\releaseL{\forma}{\formb})}) \\
           & \text{by Theorem~\ref{th:standard-ltl}} \\
           & \Leftrightarrow \releaseL{\forma}{\formb}
  \end{align*}
  \qed
\end{proof}

\subsection{Proof of Lemma \ref{le:finiteness}}

  \begin{proof}
    By straightforward induction on the linear temporal formula.
    \qed
  \end{proof}

\subsection{Proof of Lemma \ref{lemma:simp-delta-plus}}

  \begin{proof}
    By straightforward induction on the linear temporal formula.
    \qed
  \end{proof}

\subsection{Proof of Lemma \ref{le:part-deriv-direct}}

\begin{proof}
  By induction on $\forma$. % Recall that $\owedge$ is considered as $\wedge$.

  \textbf{Case} $\releaseL\forma\formb$.
    By definition,
    \bda{lcll}
    \pderiv{\releaseL\forma\formb}{x}
     & = & \{ \forma' \owedge \formb' \mid \Angle{\monoa, \forma'} \in \LF (\forma),
                  \Angle{\monob, \formb'} \in \LF (\formb), x \models \monoa \scm \monob \}  & (1)
                  \\ && \cup \{ \formb' \owedge \releaseL{\forma}{\formb} \mid \Angle{\monob,\formb'} \in \LF (\formb), x \models \monob \} & (2)
    \eda
    Consider (1).
    For $\monoa \scm \monob = \false$, the second components of the respective linear forms can be ignored.
    Hence, by IH we find that
    $\{ \forma' \owedge \formb' \mid \Angle{\monoa, \forma'} \in \LF (\forma),
    \Angle{\monob, \formb'} \in \LF (\formb), x \models \monoa \scm \monob \}
    \subseteq \{ \forma' \wedge \formb' \mid \forma' \in \partDeriv{\forma}{x}, \formb' \in \partDeriv{\formb}{x} \}$.
    The other direction follows as well as $x \models \monoa$ and $x \models \monob$ implies that
    $\monoa \scm \monob \not= \false$.
    Consider (2). By IH we have that 
    $\{ \formb' \owedge \releaseL{\forma}{\formb} \mid \Angle{\monob,\formb'} \in \LF (\formb), x \models \monob \}
    =  \{ \formb' \wedge \releaseL{\forma}{\formb} \mid \formb' \in \partDeriv{\formb}{x} \}$.
    Hence, $\pderiv{\releaseL\forma\formb}{x} = \partDeriv{\releaseL\forma\formb}{x}$.

    The other cases can be proven similarly.
\end{proof}

\subsection{Proof of Lemma \ref{le:part-deriv-cardinality}}

\begin{proof}
  The cardinality of $\partial^+ (\forma)$ is bounded by $O(n)$.
  By Lemma~\ref{le:closedness} (second part) elements in the set of descendants
  are in the set $\SET(\partial^+ (\forma))$.
  The mapping $\SET$ builds all possible (conjunctive) combinations of the underlying set.
  Hence, the cardinality of $\SET(\partial^+ (\forma))$ is bounded by $O(2^n)$ and we are done.
\end{proof}

\subsection{Proof of Lemma \ref{le:closedness}}

\begin{proof}
  \textbf{First part.} By induction on $\forma$ we show that $\{ \forma' \mid \Angle{\monoa,\forma'} \in \LF (\forma) \}
  \subseteq \SET(\partial^+ (\forma))$.

  \textbf{Case }$\true$.
  $\LF( \true) = \{ \Angle{\true, \otop} \}$ and $\otop \in \SET(\partial^+ (\true))$.

  \textbf{Case }$\Lit$. Analogous.

  \textbf{Case }$\false$. Holds vacuously.

  \textbf{Case }$\forma \vee \formb$. Immediate by induction.

  \textbf{Case }$\forma \wedge \formb$. Immediate by induction.

  \textbf{Case }$\nextL\forma$. $\LF (\nextL\forma) = \{ \Angle{\true, \forma'} \mid \forma' \in \SIMP (\forma)\}$ and
  by Lemma~\ref{lemma:simp-delta-plus}, $\SIMP (\forma) \subseteq \SET(\partial^+ (\forma))$.

  \textbf{Case }$\untilL\forma\formb$.
  $\LF (\untilL\forma\formb) = \LF (\formb) \cup \{
  \Angle{\monoa,\forma' \owedge \untilL{\forma}{\formb}} \mid
  \Angle{\monoa,\forma'} \in \LF (\forma) \}$.
  By induction, the second components of $\LF (\formb)$ are in $\SET (\partial^+ (\formb)) \subseteq \SET (\partial^+
  (\untilL\forma\formb))$.
  By induction, the second components $\forma'$ of $\LF (\forma)$ are in $\SET (\partial^+ (\forma))$, so that $\forma'
  \owedge \untilL\forma\formb \in \SET (\partial^+ (\forma) \cup \{ \untilL\forma\formb \}) \subseteq \SET (\partial^+
  (\untilL\forma\formb))$. 

  \textbf{Case }$\releaseL\forma\formb$. $\LF (\releaseL{\forma}{\formb}) =   \{  \Angle{\monoa \scm \monob, \forma' \owedge \formb'} \mid
  \Angle{\monoa, \forma'} \in \LF (\forma),
  \Angle{\monob, \formb'} \in \LF (\formb)
  \} \cup \{
  \Angle{\monob, \formb' \owedge \releaseL{\forma}{\formb}} \mid
  \Angle{\monob,\formb'} \in \LF (\formb)\}$.
  By induction $\forma' \in \SET (\partial^+ (\forma))$ and $\formb' \in \SET (\partial^+ (\formb))$ so that
  $\forma'\owedge\formb' \in \SET (\partial^+ (\forma) \cup \partial^+ (\formb)) \subseteq \SET (\partial^+
  (\releaseL\forma\formb))$. 
  Furthermore, $\formb' \owedge \releaseL\forma\formb \in \SET (\partial^+ (\formb) \cup \{ \releaseL\forma\formb \})
  \subseteq \SET (\partial^+ (\releaseL\forma\formb))$.

  \textbf{Second part.} By induction on $\forma$.

  \textbf{Case }$\Lit$. If $\forma' = \Lit$ or $\forma' = \otop$, then $\otop \in \SET (\partial^+ (\Lit))$.

  \textbf{Case }$\true$. Analogous.

  \textbf{Case }$\false$. Vacuously true.

  \textbf{Case }$\forma\vee\formb$. Immediate by induction.

  \textbf{Case }$\forma\wedge\formb$. Immediate by induction.

  \textbf{Case }$\untilL\forma\formb$. By induction and the first part.

  \textbf{Case }$\releaseL\forma\formb$. By induction and the first part.

  % We need the auxilliary statement that $\forma_1, \forma_2 \in \partial^+ (\forma)$ implies $\forma_1 \owedge \forma_2 \in \partial^+ (\forma)$. \textbf{MS seems to hold, CHECK}

  % \textbf{Case}$\alwaysL\forma$.
  % Consider $\forma' \owedge \alwaysL\forma \in \partial^+ (\alwaysL\forma)$ where $\forma' \in \partial^+ (\forma)$.
  % We observe
  % \bda{l}
  % \pderiv{\alwaysL\forma}{x}
  % \\ = \{ \forma' \owedge \alwaysL\forma \mid \Angle{\monoa,\forma'} \in \LF (\forma) \}
  % \\ \subseteq \{ \forma' \owedge \alwaysL\forma \mid \forma' \in \partial^+ (\forma) \}
  % \eda

  % Thus, we find that
  
  % \bda{l}
  % \pderiv{\forma' \owedge \alwaysL\forma}{x}
  % \\ = \{ \forma'' \owedge \forma''' \mid \forma'' \in \pderiv{\forma'}{x}, \forma''' \in \pderiv{\alwaysL\forma}{x} \}
  % \\ \text{by the above observation}
  % \\ \subseteq \{ \forma'' \owedge \forma'''' \owedge \alwaysL\forma \mid \forma'' \in \pderiv{\forma'}{x}, \forma'''' \in \partial^+ (\forma) \}
  % \\ \text{by IH}
  % \\ \subseteq \{ \forma'' \owedge \forma'''' \owedge \alwaysL\forma \mid \forma'' \in \partial^+ (\forma), \forma'''' \in \partial^+ (\forma) \} \\ \text{by the auxilliary statement}
  % \\ = \{ \forma'' \owedge \alwaysL\forma \mid \forma'' \in \partial^+ (\forma) \}
  % \\ = \partial^+ (\alwaysL\forma)
  % \eda
  % \textbf{MS, remaining cases}
  % ~\qed
\end{proof}

\subsection{Proof of Theorem \ref{th:ltl-to-aa}}

\begin{proof}
  Suppose that $\sigma \models \forma$. Show by induction on $\forma$ that $\sigma \in \L (\A
  (\forma))$.

  \textbf{Case }$\true$. Accepted by run $\true, \otop, \dots$ which visits $\otop \in F$ infinitely often.

  \textbf{Case }$\false$. No run.

  \textbf{Case }$p$. As $p\in \INT (\sigma_0)$, $\sigma$ is accepted by run $p, \otop, \otop, \dots$.

  \textbf{Case }$\neg p$. Accepted by run $\neg p, \otop, \otop, \dots$.

  \textbf{Case }$\forma \wedge \formb$.
  By definition $\sigma \models \forma$ and $\sigma \models \formb$.
  By induction, there are accepting runs $\alpha_0, \alpha_1, \dots$ on $\sigma$ in $\A (\forma)$
  and $\beta_0, \beta_1, \dots$ on $\sigma$ in $\A (\formb)$.
  But then $\alpha_0 \owedge \beta_0, \alpha_1 \owedge \beta_1, \dots$ is an accepting run on
  $\sigma$ in $\A (\forma \wedge \formb)$ \textbf{because the state sets of the automata are
    disjoint}. 

  \textbf{Case }$\forma \vee \formb$.
  By definition $\sigma \models \forma$ or $\sigma \models \formb$.
  If we assume that $\sigma \models \forma$, then induction yields an accepting run $\alpha_0,
  \alpha_1, \dots$ on $\sigma$ in $\A (\forma)$.
  As the initial state of $\A (\forma \vee \formb)$ is chosen from $\{\alpha_0, \beta_0 \}$, for some $\beta_0$,
  we have that $\alpha_0, \alpha_1, \dots$ is an accepting run on $\sigma$ in $\A
  (\forma\vee\formb)$.

  \textbf{Case }$\nextL\forma$.
  By definition $\sigma[1\dots] \models \forma$.
  By induction, there is an accepting run $\alpha_0, \alpha_1, \dots$ on $\sigma[1\dots]$ in $\A
  (\forma)$ with $\alpha_0 = \SIMP (\forma)$.
  Thus, there is an accepting run $\nextL\forma, \alpha_0, \alpha_1, \dots$ on $\sigma$ in $\A
  (\nextL\forma)$.

  \textbf{Case }$\untilL\forma\formb$.
  By definition $\exists n \in \nat, \forall j \in \nat, j<n \Rightarrow \sigma[j\dots] \models
  \forma$ and $\sigma[n\dots] \models \formb$. By induction, there is an accepting run on
  $\sigma[n\dots]$ in $\A (\formb)$ and, for all $0 \le j<n$, there are accepting runs on
  $\sigma[j\dots]$ in $\A (\forma)$.

  We proceed by induction on $n$.

  \textbf{Subcase }$n=0$. In this case, there is an accepting run $\beta_0, \beta_1, \dots$ on
  $\sigma[0\dots] = \sigma$ in $\A (\formb)$ so that $\beta_0 = \SIMP (\formb)$. We want to show
  that $\untilL\forma\formb, \beta_1, \dots$ is an accepting run on $\sigma$ in $\A
  (\untilL\forma\formb)$. 
  To see this, observe that $\beta_1 \in \pderiv{\beta_0}{\sigma_0}$
  and that $\pderiv{\untilL\forma\formb}{\sigma_0} = \pderiv{\beta_0}{\sigma_0} \cup \pderiv{\alpha_0}{\sigma_0}
  \owedge \untilL\forma\formb$, where $\alpha_0 = \SIMP (\forma)$, which proves the claim.
  % \textbf{MS: $\ovee$ doesn't actually arise as $\pderiv{}{}$ yields a set of states,
  %   or do we assume that $\pderiv{\forma}{x} = \{\forma_1,\dots,\forma_n\} = \forma_1 \ovee \dots \ovee \forma_n$.}

  \textbf{Subcase }$n>0$. There must be an accepting run $\alpha_0, \alpha_1, \dots$ on
  $\sigma[0\dots] = \sigma$ in $\A (\forma)$ so that $\alpha_0 = \SIMP (\forma)$. By induction (on
  $n$) there must be an accepting run $\beta_0, \beta_1, \dots$ on $\sigma[1\dots]$ in $\A
  (\untilL\forma\formb)$ where $\beta_0 = \untilL\forma\formb$. We need to show that $\untilL\forma\formb,
  \alpha_1 \owedge \beta_0,  \alpha_2 \owedge
  \beta_1, \dots$ is an accepting run on $\sigma$ in $\A (\untilL\forma\formb)$.
  By the analysis in the base case, the automaton can step from $\untilL\forma\formb$ to
  $\pderiv{\alpha_0}{\sigma_0} \owedge \untilL\forma\formb$.

  \ms{MS: probably it's obvious that should mention that final states $\A (\formb)$ are contained in final states $\A (\untilL{\forma}{\formb})$.}

  \ms{MS: notation, a run is a digraph, we should be fine as we only change the starting vertex of the graph.}
  
  \textbf{Case }$\releaseL\forma\formb$.
  % \textbf{PT: CHECK!!}

  By definition, $\forall n \in \nat,
    ( \sigma[n\dots] \models \formb
    \text{ or }
    \exists j \in \nat,
    ((j<n) \wedge  \sigma[j\dots] \models \forma))
    $.
    By induction, there is either an accepting run on $\sigma[n\dots]$
    in $\A (\formb)$, for each $n\in\nat$, or there exists some
    $j\in\nat$ such that there is an accepting run on $\sigma[j\dots]$
    in $\A (\forma)$ and for all $0\le i\le j$, there is an accepting run
    on $\sigma[i\dots]$ in $\A (\formb)$.

    \ms{MS: notation, how are $E$ and $\pi$ connected?}
    If there is an accepting run $\pi_0^n, E_0^n, \pi_1^n, E_1^n, \dots$ in $\A (\formb)$ on
    $\sigma[n\dots]$ for each $n\in\nat$ where $\pi_0^n \in \SIMP
    (\formb)$ and $\pi_{i+1}^n \in \pderiv{\pi_i^n}{\sigma_{i+n}}$, then there is an accepting
    run in $\A (\releaseL\forma\formb)$:

    $\pderiv{\releaseL\forma\formb}{\sigma_0} =
    \pderiv{\forma\owedge\formb}{\sigma_0} \cup
    \pderiv{\formb}{\sigma_0} \owedge \releaseL\forma\formb$.

    \ms{MS: don't follow the below, this should be the case ``there exists $j$ ...''.}
    Suppose that  there is either an accepting run on $\sigma[n\dots]$
    in $\A (\formb)$, for each $n\in\nat$. In this case, there is an accepting run in $\A (\releaseL\forma\formb)$:
    there is infinite path of accepting states $\releaseL\forma\formb, \dots$ and, as $\formb$ holds at every $n$, every
    infinite path that starts in a state in $\pderiv{\formb}{\sigma_n}$ visits infinitely many accepting states.

    Otherwise, the run visits only finitely many states of the form  $\releaseL\forma\formb$ and then continues
    according to the accepting runs on $\forma$ and $\formb$ starting with
    $\pderiv{\forma\owedge\formb}{\sigma_j}$. Furthermore, any infinite path starting at some $\pderiv{\formb}{\sigma_i}
    \owedge \releaseL\forma\formb$ that goes through $\pderiv{\formb}{\sigma_i}$ visits infinitely many accepting states
    (for $0\le i<j$).

    \textbf{Suppose now} that $\sigma\not\models \forma$ and show that
    $\sigma\notin\L (\A (\forma))$.

    $\sigma\not\models\forma$ is equivalent to
    $\sigma\models\neg\forma$. We prove by induction on
    $\forma$ that $\sigma \notin \L (\A (\varphi))$.

    \textbf{Case }$\true$. The statement $\sigma\not\models\true$ is
    contradictory.

    \textbf{Case }$\false$. The statement $\sigma\not\models\false$
    holds for all $\sigma$ and the automaton $\A (\false)$ has no
    transitions, so $\sigma \notin \L (\A (\false))$.

    \textbf{Case }$p$. The statement $\sigma\not\models p$ is
    equivalent to $\sigma\models \neg p$.  That is, $p \notin
    \INT(\sigma_0)$. As $\LF (p) = \{ \Angle{p, \true} \}$, we find
    that $\pderiv{p}{\sigma_0} = \emptyset$ so that $\A (p)$ has no
    run on $p$.

    \textbf{Case }$\neg p$. Similar.

    \textbf{Case }$\forma\wedge\formb$. If
    $\sigma\not\models\forma\wedge\formb$, then
    $\sigma\not\models\forma$ or $\sigma\not\models\formb$.
    If we assume that $\sigma\not\models\forma$ and appeal to
    induction, then either there is no run of $\A (\forma)$ on $\sigma$: in this case, there is no run
    of $\A (\forma\wedge\formb)$ on $\sigma$, either.
    Alternatively, every run of $\A (\forma)$ on $\sigma$ has a path
    with only finitely many accepting states. This property is
    inherited by $\A (\forma\wedge\formb)$.

    \textbf{Case }$\forma\vee\formb$. If 
    $\sigma\not\models\forma\vee\formb$, then
    $\sigma\not\models\forma$ and $\sigma\not\models\formb$. By
    appeal to induction, every run of $\A (\forma)$ on $\sigma$ as
    well as every run of $\A (\formb)$ on $\sigma$ has a
    path with only finitely many accepting states. Thus, every run of
    $\A (\forma\vee\formb)$ on $\sigma$ will have an infinite path with only
    finitely many accepting states.

    \textbf{Case }$\nextL\forma$. If $\sigma\not\models\nextL\forma$,
    then $\sigma\models\neg\nextL\forma$ which is equivalent to
    $\sigma\models\nextL{\neg\forma}$ and thus
    $\sigma[1\dots]\not\models\forma$. By induction every run of $\A
    (\forma)$ on $\sigma[1\dots]$ has an infinite path with only finitely many
    accepting states, so has every run of $\A (\nextL\forma)$ on
    $\sigma$.

    \textbf{Case }$\untilL\forma\formb$. If $\sigma\not\models
    \untilL\forma\formb$, then it must be that
    $\sigma\models\releaseL{(\neg\forma)} {(\neg\formb)}$.

    By definition, the release formula holds if
    \begin{gather*}
      \forall n\in\nat, (\sigma[n\dots] \not\models\formb \text{ or
      }\exists j\in\nat, (j< n \wedge \sigma[j\dots] \not\models\forma))
    \end{gather*}
    We obtain, by induction, for all $n\in\nat$ that either
    \begin{enumerate}
    \item\label{item:1} every run of $\A (\formb)$ on $\sigma[n\dots]$ has an infinite path
      with only finitely many accepting states or
    \item\label{item:2} $\exists j\in\nat$ with $j<n$ and every run of $\A (\forma)$
      on $\sigma[j\dots]$ has an infinite path with only finitely many
      accepting states.
    \end{enumerate}
    Now we consider a run of $\A (\untilL\forma\formb)$ on $\sigma$. 
    \begin{align*}
      \pderiv{\untilL\forma\formb}{\sigma_0}
      &= \{ \forma' \mid \Angle{\monoa, \forma'} \in \LF (\untilL\forma\formb), \sigma_0 \models \monoa \} \\
      &= \{ \formb' \mid \Angle{\monob, \formb'} \in \LF (\formb) , \sigma_0 \models \monob \} \\
      & \cup \{  \forma' \owedge \untilL\forma\formb \mid \Angle{\monoa, \forma'} \in \LF (\forma) , \sigma_0 \models \monoa \} 
    \end{align*}
    To be accepting, the run cannot always choose the alternative that contains $\untilL\forma\formb$ because that would
    give rise to an  infinite path $(\untilL\forma\formb)^\omega$ which contains no accepting state.

    Thus, any accepting run must choose the alternative containing $\formb'$ a derivative of $\formb$. Suppose this
    choice happens at $\sigma_i$. If the release formula is accepted because case~\ref{item:1} holds always, then a run
    of $\A (\formb)$ starting at $\sigma_i$ has an infinite path with only finitely many accepting states. So this run
    cannot be accepting.

    If the release formula is accepted because eventually case~\ref{item:2} holds, then $i<j$ is not possible for the
    same reason as just discussed. However, starting from $\sigma_j$, we have a state component from $\A (\forma)$ which
    has an infinite path with only finitely many accepting states. So this run cannot be accepting, either.

    \textbf{Case }$\releaseL\forma\formb$. If $\sigma\not\models
    \releaseL\forma\formb$, then $\sigma\models\neg
    (\releaseL\forma\formb)$ which is equivalent to 
    $\sigma\models\untilL{(\neg\forma)} {(\neg\formb)}$.

    By definition, the until formula holds if
    \begin{gather*}
      \exists n \in \nat,
      (\forall j \in \nat, j<n \Rightarrow \sigma[j\dots] \not\models \forma)
      \text{ and } \sigma[n\dots] \not\models \formb
    \end{gather*}
    We obtain, by induction, that there is some $n\in\nat$ such that
    \begin{enumerate}
    \item\label{item:3} for all $j\in\nat$ with $j<n$ every run of $\A (\forma)$ on $\sigma[j\dots]$ has an infinite path with only
      finitely many accepting states and
    \item\label{item:4} every run of $\A (\formb)$ on $\sigma[n\dots]$ has an infinite path with only finitely many accepting states. 
    \end{enumerate}
    Now we assume that there is an accepting run of $\A (\releaseL\forma\formb)$ on $\sigma$. Consider
    \begin{align*}
      \pderiv{\releaseL\forma\formb}{\sigma_0}
      &=     \pderiv{\forma\owedge\formb}{\sigma_0} \cup
        \pderiv{\formb}{\sigma_0} \owedge \releaseL\forma\formb
    \end{align*}
    Suppose that the run always chooses the alternative containing the formula $\releaseL\forma\formb$. However, at
    $\sigma_n$, this formula is paired with a run of $\A (\formb)$ on $\sigma[n\dots]$ which has an infinite path with
    only finitely many accepting states. A contradiction.

    Hence, there must be some $i\in\nat$ such that $\A (\releaseL\forma\formb)$ chooses its next states from
    $\pderiv{\forma\owedge\formb}{\sigma_i}$. If this index $i<n$, then this run cannot be accepting because it contains
    a run of $\A (\forma)$ on $\sigma[i\dots]$, which has an infinite path with only finitely many accepting
    states. Contradiction.

    On the other hand, $i\ge n$ is not possible either because it would contradict case~\ref{item:4}.

    Hence, there cannot be an accepting run.
    \qed
\end{proof}

%-------------------------------------------------------------------------  
\subsection{Proof of Theorem~\ref{th:correctness-optimized-variant}}

We observe that exhaustive decomposition yields to the same set of states,
regardless of the order decomposition rules are applied.

\begin{example}
  Consider $\alwaysL{p} \wedge \eventuallyL{\neg p}$.
  Starting with $\{ \{ \alwaysL{p} \wedge \eventuallyL{\neg p} \} \}$
  the following rewrite steps  can be applied. Individual rewrite steps
  are annotated with the decomposition rule (number)
  that has been applied.
  \bda{l}
  \alwaysL{p} \wedge \eventuallyL{\neg p}
   \ \stackrel{2}{\rightarrowtail} \ \{ \{ \alwaysL{p}, \eventuallyL{\neg p} \} \}
   \ \stackrel{4}{\rightarrowtail} \ \{ \{ p, \nextL\alwaysL{p}, \eventuallyL{\neg p} \} \}
   \ \stackrel{3}{\rightarrowtail} \ \{ \{ p, \nextL\alwaysL{p}, \neg p \}, \{ p, \nextL\alwaysL{p}, \nextL\eventuallyL{\neg p} \} \}
  \eda
  In the final set of nodes we effectively find nodes $S_3$ and $S_4$
  from Wolper's tableau construction. Intermediate nodes $S_1$ and $S_2$ arise in some intermediate rewrite steps.
  See Figure~\ref{fig:tableau}.
  The only difference is that marked formulae are dropped.

  An interesting observation is that there is an alternative rewriting.
  \bda{ll}
  \alwaysL{p} \wedge \eventuallyL{\neg p}
  \ \stackrel{2}{\rightarrowtail} \ \{ \{ \alwaysL{p}, \eventuallyL{\neg p} \} \}
  \ \stackrel{3}{\rightarrowtail} \ \{ \{ \alwaysL{p}, \neg p \}, \{ \alwaysL{p}, \nextL\eventuallyL{\neg p} \} \}
  & \stackrel{4}{\rightarrowtail} \ \{ \{ p, \nextL\alwaysL{p}, \neg p  \}, \{ \alwaysL{p}, \nextL\eventuallyL{\neg p} \} \}
  \\ & \stackrel{4}{\rightarrowtail} \ \{ \{ p, \nextL\alwaysL{p}, \neg p \}, \{ p, \nextL\alwaysL{p}, \nextL\eventuallyL{\neg p} \} \}
  \eda
  The set of children nodes reached is the same. 
\end{example}

We formalize the observations made in the above example.
Decomposition yields the same set of nodes regardless if there is a choice in some intermediate step.

\begin{lemma}
 \label{le:rewrite-cnf} 
  The rewrite relation $\rightarrowtail$ is terminating and confluent.
\end{lemma}
\begin{proof}
By inspection of the decomposition rules D1-6.
\end{proof}

So, our reformulation of Wolper's tableau construction
method yields the same nodes (ignoring marked formulae and intermediate nodes).

\begin{lemma}
  Let $S$ be a pre-state node in Wolper's tableau construction
  and $S'$ be a node derived from $S$ via some (possibly repeated) decomposition steps
  where $S'$ is a state.
  Then, $\{ S \} \rightarrowtail^* N$ for some $N$ where $S'' \in N$ such that
  $S''$ and $S'$ are equivalent modulo marked formulae.
\end{lemma}
\begin{proof}
  On a state no further decomposition rules can be applied.
  The only difference between our rewriting based formulation
  of Wolper's tableau construction is that we drop marked formulae.
  Hence, the result follows immediately.
\end{proof}

Wolper's proof does not require marked formulae not does he make
use of intermediate nodes in any essential way.
\ms{Wolper says he requires marked formulae but they don't show up in the proof.}
Hence, we argue that correctness of the optimized Wolper-style
tableau construction method follows from Wolper's proof.

%-------------------------------------------------------------------------  
\subsection{Proof of Lemma~\ref{le:state-vs-lf}}

We first state some auxiliary result.

\begin{lemma}
\label{le:exhaust-decompose-single}  
  Let $\{ S \cup \{\forma\} \} \cup N \rightarrowtail \{ S \cup S_1 \} \cup \dots \cup \{ S \cup S_n \} \cup N
  \rightarrowtail^* N'$
  where $\forma \rightarrow \{ S_1, \dots, S_n \}$
  and $\{ \{ \forma \} \} \rightarrowtail^* \{ S_1', \dots, S_m' \}$.
  Then, $\{ S \cup \{\forma\} \} \cup N \rightarrowtail \{ S \cup S'_1 \} \cup \dots \cup \{ S \cup S'_m \} \cup N
         \rightarrowtail^* N'$.
\end{lemma}
\begin{proof}
  By induction over the length of the derivation $\{ \{ \forma \} \} \rightarrowtail^* \{ S_1', \dots, S_m' \}$
  and the fact that the rewriting relation is terminating and confluent (Lemma~\ref{le:rewrite-cnf}).
\end{proof}

The above says that we obtain the same result, regardless,
if we exhaustively decompose a single formula, or apply 
decomposition steps that alternate among multiple formulae.
This simplifies the up-coming inductive proof of Lemma~\ref{le:rewrite-cnf}.

By induction on $\forma$ we show that if $\forma \rightarrowtail^* N$ then $\LF (\forma) = \sem{N}$.
\begin{proof}
  \textbf{Case }$\forma \wedge \formb$.
     By assumption $\forma \wedge \formb \rightarrowtail \{ \{\forma, \formb\} \} \rightarrowtail^* N$.
     By induction we find that (1) $\LF (\forma) = \sem{N_1}$ and (2) $\LF (\formb) = \sem{N_2}$
     where $\forma \rightarrowtail^* \{ S_1,\dots, S_n \}$, $\formb \rightarrowtail^* \{ T_1,\dots, T_m \}$,
     $N_1 = \{ S_1,\dots, S_n \}$ and $N_2 = \{ T_1,\dots, T_m \}$.
     By Lemma~\ref{le:exhaust-decompose-single}, we can conclude that
     $\forma \wedge \formb \rightarrowtail \{ \{ \formb \} \cup S_1, \dots, \{ \formb \} \cup S_n \}
     \rightarrowtail \{ S \cup T \mid S \in \{S_1,\dots,S_n\}, T \in \{T_1,\dots,T_m \}\}$
     where $N = \{ S \cup T \mid S \in \{S_1,\dots,S_n\}, T \in \{T_1,\dots,T_m \}\}$.
     From this and via (1) and (2), we can derive that
     $\LF (\forma \wedge \formb) = \sem{N}$. Recall that elimination via E1 is integrated
     as part of rewriting (see Definition~\ref{def:decomp-elim}).

 \textbf{Case}$\releaseL{\forma}{\formb}$.
    By assumption $\releaseL{\forma}{\formb} \rightarrowtail \{ \{ \formb, \forma \vee \nextL{(\releaseL{\forma}{\formb})} \} \} \rightarrowtail \{ \{ \formb, \forma \}, \{ \formb, \nextL{(\releaseL{\forma}{\formb})} \} \}
            \rightarrowtail^* N$.
    Via pretty much the same reasoning as in case of conjunction, we can establish
    that $\LF (\releaseL{\forma}{\formb}) = \sem{N}$

    For brevity, we skip the remaining cases. They all follow the same pattern.
\end{proof}

\end{document}